\documentclass[aps,prl,reprint,superscriptaddress,floatfix,nofootinbib,amsmath,amssymb]{revtex4-2}

\usepackage{graphicx}
\usepackage{dcolumn}
\usepackage{bm}
\usepackage{hyperref}
\usepackage{amsthm} 
\usepackage{amsmath} 
\usepackage{amssymb} 
\usepackage{physics} 
\usepackage{enumitem}
\usepackage{textcomp}
\usepackage{float}
\usepackage{capt-of}  
\usepackage[capitalize]{cleveref}
\usepackage[caption=false]{subfig} 

\newtheorem{theorem}{Theorem}

\theoremstyle{definition}
\newtheorem{lemma}{Lemma}

\newtheorem{definition}{Definition}

\begin{document}

\title{Adiabatic Encoding of Pre-trained MPS Classifiers into Quantum Circuits}

\author{Keisuke Murota}
\affiliation{Department of Physics, The University of Tokyo, Tokyo 113-0033, Japan}

\date{\today}

\begin{abstract}
Although Quantum Neural Networks (QNNs) offer powerful methods for classification tasks,
the training of QNNs faces two major training obstacles: barren plateaus and local minima.
A promising solution is to first train a tensor-network (TN) model classically and then embed it into a QNN.\
However, embedding TN-classifiers into quantum circuits generally requires 
postselection whose success probability may decay exponentially with the system size. 
We propose an \emph{adiabatic encoding} framework that encodes pre-trained 
MPS-classifiers into quantum MPS (qMPS) circuits with postselection, and gradually 
removes the postselection while retaining performance.
We prove that training qMPS-classifiers from scratch on a certain artificial dataset is 
exponentially hard due to barren plateaus, but our adiabatic encoding circumvents this issue. 
Additional numerical experiments on binary MNIST also confirm its robustness.
\end{abstract}

\maketitle

\textit{Introduction.---}
Due to the huge success of deep neural networks, 
quantum neural networks (QNNs) have recently attracted significant attention~\cite{Beer2020,Biamonte2017,Schuld2014,Mitarai2018}.
Some works have indeed reported 
possible advantages of QNNs in terms of expressibility~\cite{Sim2019,Beer2020,Abbas2021}, generalizability~\cite{Caro2022,Huang2022}, and 
computational power~\cite{Liu2021,Havlicek2019,Yamasaki2023}. 
However, their optimization remains notoriously challenging. 
One key obstacle is the so-called barren plateaus, in which gradients 
vanish exponentially with the number of qubits
unless one provides a suitable initialization~\cite{McClean2018,Cerezo2021,Wang2021,Grant2019}.
In addition, even if the barren plateaus do not exist, local minima can trap 
the optimization in a poor loss landscape~\cite{You2021,Anschuetz2022,Bittel2021,Cerezo2022}.
To this end, recent works have introduced a synergistic
framework that first pre-train TN models such as 
matrix product states (MPS) or tree tensor networks (TTN)  
to obtain a high-quality solution, which is then encoded into quantum circuits, 
sometimes called as quantum TN (qTN)-circuits. 
After preparing high-quality qTN-circuits by embedding classically pre-trained TN models, 
additional quantum gates can be introduced to reach beyond classically simulatable regimes and thereby further boost 
performance~\cite{Rudolph2023,Iaconis2024,Wall2021,Rudolph2024,Ran2020,Malz2024,Sugawara2025}. 
This two-step strategy has proven effective for tasks such as variational 
quantum eigensolvers (VQEs)~\cite{Peruzzo2014}, advanced state preparation and generative tasks,
and it holds promise for realizing QNNs that outperform state-of-the-art 
classical algorithms.

Extending this synergetic framework to classification tasks appears equally natural and promising: 
one might classically train TN-classifiers~\cite{Stoudenmire2016, Novikov2016, Chen2024, Efthymiou2019, Mossi2024} and then encode them into qTN-circuits
that can be served as promising QNN models for classification.
However, a closer look reveals a critical hindrance: 
strictly realizing the TN-classifiers as QNN models requires postselection on intermediate measurement outcomes.
This introduces severe overhead, often exponentially as the system size grows~\cite{Kodama2022}. 
Therefore, no matter how powerful the classically trained TN-classifiers may be, 
the overhead imposed by postselection effectively prevents us
from harnessing any quantum advantage when directly embedding them into quantum circuits.
There are past efforts to implement TN-classifiers as qTN-circuits~\cite{Rieser2023,Kundu2024,Wall2021-TTN,Wall2022}.
However, they typically proceeded in the hope that 
future developments would resolve the overhead, or attempted partial workarounds with limited success to date.

In this work, we focus on MPS as the most representative TN architecture
and propose an \emph{adiabatic encoding} framework that seamlessly 
converts MPS-classifiers into postselection-free quantum MPS-classifiers (qMPS-classifiers). 
Our key insight is that both the postselection-based and postselection-free circuits 
can be understood as special cases of ``weighted qMPS'' ansätze.
This recognition enables a smooth transition away from postselection to postselection-free circuits while
maintaining the high performance of pre-trained MPS-classifiers.
Although we illustrate our approach using MPS, 
the underlying principles generalize to other TN architectures.
Importantly, this framework not only alleviates the challenges associated with training qTN-circuits for classification, 
but also opens new pathways toward quantum advantage in classification
and other supervised learning tasks.

\textit{Preliminaries.---}
\begin{figure}[tbp]
	\centering
	\includegraphics[width=0.48\textwidth]{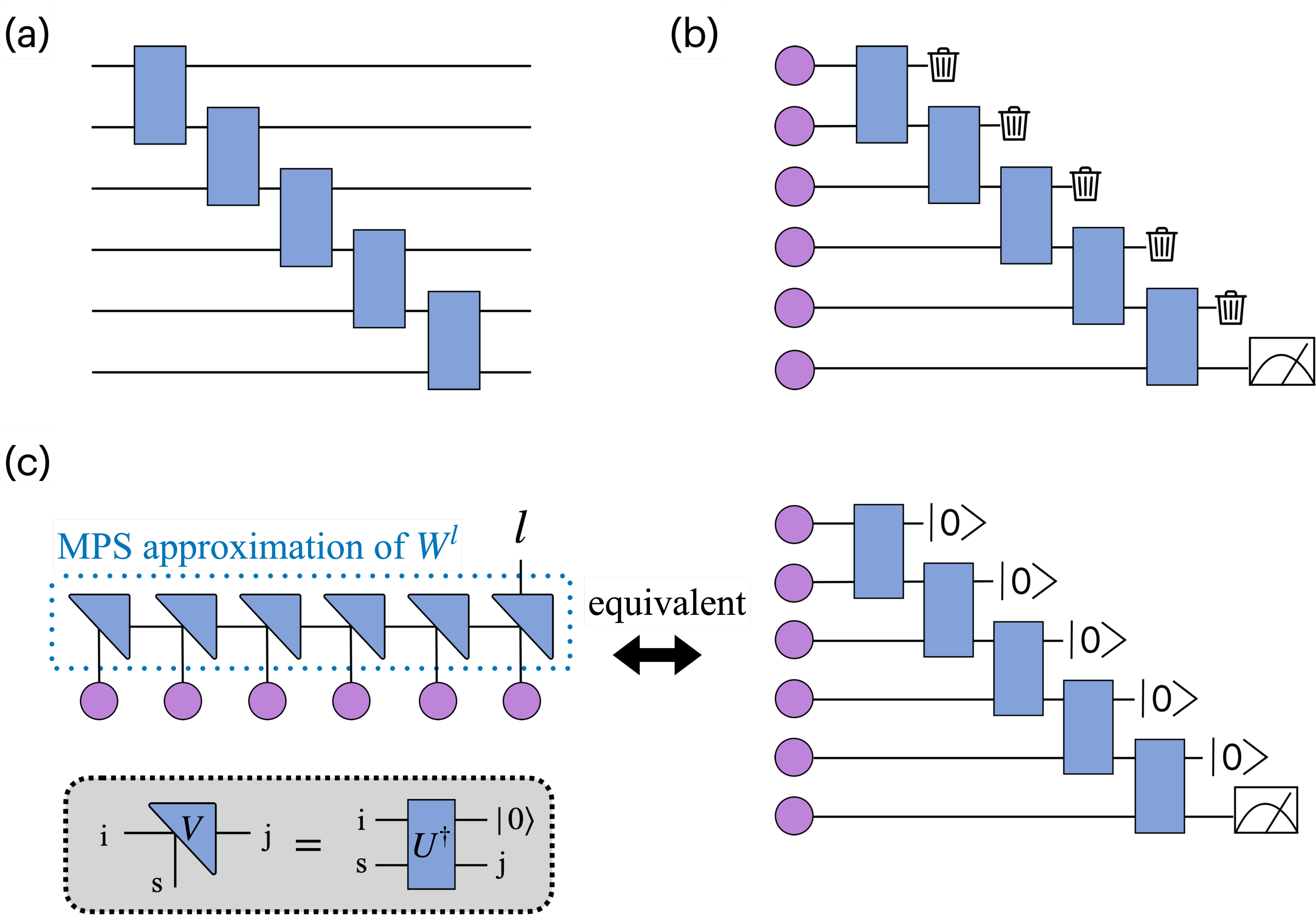}
	\caption{(a) A qMPS-circuit with 2-qubit gates. 
    (b) A qMPS-classifier without postselection. The purple circles represent the embedded input data $\phi(\vec{x})$, rectangles represent 2-qubit gates, and the trash bins indicate the discarding of the measurement outcomes.
    (c) Correspondence between MPS-classifiers and qMPS-classifiers with postselection. 
    The triangle is the isometry $V$.
    $\vert 0 \rangle$ indicates the outcome on which we perform postselection.
    }
	\label{fig:qMPS_circuit}
\end{figure}
We begin by introducing qMPS-circuits denoted by $U_{\text{qMPS}}$.
As illustrated in~\cref{fig:qMPS_circuit} (a), qMPS-circuits have a layered structure where unitary gates are arranged diagonally, 
so that it can preserve the underlying MPS representation.
The rank 3 tensors with bond dimension $\chi$ in the MPS translates into $\log_2(\chi) + 1$ qubit unitary gates in the qMPS-circuits.
For simplicity, we only consider the case of $\chi=2$ and 2-qubit gates in this work.
This means that extending $U_{\text{qMPS}}$ to larger $\chi$ to achieve higher expressibility 
is straightforward if we are allowed to use multi-qubit gates.
We want to emphasize that research on qMPS-circuits itself is also rapidly progressing~\cite{Foss2021,Niu2022,Huang2023}, 
with growing insights into the expressibility of qMPS-circuits compared to classical MPS~\cite{Araz2022,Haghshenas2022}, 
and the emergence of barren plateaus when attempted to train from random initializations~\cite{Zhao2021,Cervero2023,Liu2022}.

We can use these qMPS-circuits for classification tasks~\cite{Huggins2019}.
The process can be outlined in several steps:
(i) First, the input classical data $\vec{x} \in \mathbb{R}^L$
is embedded into a high-dimensional quantum state 
$|\phi(\vec{x}) \rangle = |\phi(x_1)\rangle \otimes |\phi(x_2)\rangle \otimes \cdots \otimes |\phi(x_L)\rangle \in \mathbb{C}^{2^L}$.
The feature map $\phi(\cdot)$ encodes the real number $x$ into a single qubit state, which can be chosen arbitrarily.
(ii) Apply qMPS-circuits to get $U_{\text{qMPS}}|\phi(\vec{x})\rangle$.
(iii) We measure the final qubit in the computational basis and estimate 
\begin{equation}
P(\ell \mid \vec{x}) = \langle \phi(\vec{x})\mid U_{\text{qMPS}}^\dagger \vert \ell \rangle \langle \ell \vert U_{\text{qMPS}} \mid \phi(\vec{x})\rangle,
\end{equation}
where $\ell$ is the binary class label $\in \{0,1\}$.
(iv) Using the estimated probability, we assign the class label of the input $\vec{x}$ to
$\text{argmax}_{\ell} \hat{P}(\ell \mid \vec{x})$.
This process can be simply described as in~\cref{fig:qMPS_circuit} (b).
2-qubit unitary gates in $U_{\text{qMPS}}$ are trainable and optimized during the learning process.
In general, the qMPS-classifiers can handle multi-class classification tasks by simultaneously measuring multiple qubits; however, we focus on the binary classification case in this work.
We note that these qMPS-classifiers can be implemented just with 2 qubits if 
quantum hardware supports mid-circuit measurement and reset\cite{Huggins2019}.
We call these qMPS-classifiers as ``qMPS-classifiers without postselection'' or simply as ``qMPS-classifiers''.

The MPS-classifiers operate on the similar principles as the qMPS-classifiers~\cite{Mossi2024,Stoudenmire2016}.
We first represent the tensor network $W^\ell$, a collection of vectors labeled by $\ell$, as a MPS.\
\begin{equation}
    \label{eq:mps_classifier}
    W^\ell_{s_1, s_2, \cdots, s_L} = \sum_{i_1, i_2, \cdots, i_L} A^{s_1}_{i_1} A^{s_2}_{i_2, i_3} \cdots A^{s_{L-1}}_{i_{L-1}, i_L} A^{s_L}_{i_L, \ell}
\end{equation}
where $A^{s_k}_{i_k, i_{k+1}}$ is the 3-rank tensor at the $k$-th site, $s_k$ corresponds to the physical index of size $2$ and
$i_k$ is the virtual index of size $\chi = 2$.
Using this tensor network representation as in~\cref{eq:mps_classifier}, 
we can describe the MPS-classifiers using the following steps:
(i) Given an input data point $\vec{x}$, it is first mapped into a high-dimensional feature space 
via the same feature map $\phi(\vec{x})$.
(ii) We then calculate the contraction of the MPS-classifiers as 
\begin{equation}\label{eq:mps_classifier_function}
    f_\ell(\vec{x}) = W^\ell\cdot\phi(\vec{x}).
\end{equation}
(iii) The predicted class label is obtained by evaluating $\text{argmax}_{\ell} (|f_\ell(\vec{x})|)$.
You may also apply the softmax function to the $f_\ell(\vec{x})$, before calculating the argmax,
to get the probability of each class but the result is the same.

Since there are strong synergies between the qMPS-classifiers and MPS-classifiers,
it appears natural to attempt to transform pre-trained MPS-classifiers into
qMPS-circuits~\cite{Rieser2023}.
As an initial step, we can rewrite the MPS-classifiers in canonical form as a chain of isometries $V$~\cite{Schollwoeck2011}.
Then, by embedding the isometry $V$ of size $4 \times 2$ into a 2-qubit unitary gate $U$ using the 
relationship $V_{\{i, s\}, j} = \langle i, s \vert U^\dagger \vert 0, j \rangle$,
we can exactly rewrite the MPS-classifiers as qMPS-circuits with postselection as illustrated in~\cref{fig:qMPS_circuit} (c).
If one is aware that block encoding of a matrix $A$
generally requires postselection~\cite{Low2019,Nibbi2023},
one can similarly recognize why embedding the MPS-classifiers
into a PQC here also requires postselection.
These classifiers with postselection will be refered to as ``qMPS-classifiers with postselection'' in the following.
Although we can exactly rewrite the MPS-classifiers using qMPS-circuits,
the probability that the measurement outcomes of the $L-1$ qubits in such circuits are all zero 
often decreases exponentially with the system size, which makes the QNN impractical.

\textit{An Artificial Dataset Inducing Exponentially Hard Barren Plateaus and Postselection.---}
We begin by illustrating the dual exponential-cost obstacles that arise for qMPS-classifiers on an artificial problem which we call \emph{First-Qubit Trigger} dataset. 
You can find proofs, numerical experiments and more details for below theorems in the Supplemental Material.

\begin{definition}[First-Qubit Trigger Dataset]
We define the first-qubit trigger dataset to be the set of data points 
$(\vec{x},\ell)$ where $\vec{x}=(\ell,x_2,\ldots,x_L)\in\{0,1\}^L$, 
which means the binary label $\ell$ is precisely the first component of $\vec{x}$.
\end{definition}

At first glance, the first-qubit trigger dataset might seem trivial. 
Indeed, it admits perfect classification parameters for both MPS-classifiers and qMPS-classifiers. 
\begin{lemma}[Existence of Perfect Classifiers for the First-Qubit Trigger Dataset]\label{lemma:perfect_classifier}
For the first-qubit trigger dataset, there exist optimal parameters for both MPS-classifiers and qMPS-classifiers.
\end{lemma}

\begin{proof}
\textbf{(MPS-classifiers)} One can set the first tensor and the last tensor to 
$A^{s_1}_{i_1} = \delta_{i_1, s_1}$ and $A^{s_L}_{i_L, \ell} = \delta_{i_L, \ell}$,
and choose all intermediate tensors $A^{s_k}_{i_k, i_{k+1}} = \delta_{i_k, i_{k+1}}$.
These constructions perfectly classify the dataset.
\textbf{(qMPS-classifiers)} One can set every two-qubit gate to be a SWAP.
\end{proof}

Nevertheless, while there exist the optimal parameters for qMPS-classifiers, 
starting from a random initialization and attempting to find 
such solutions using gradient-based training is exponentially hard for qMPS-classifiers. 

\begin{theorem}[\textbf{Exponential Indistinguishability under Haar-Random Initialization}]\label{thm:exp_decay}
Let $\rho_L(\vec{x})$ be the $2\times 2$ reduced density matrix at the final qubit (right before the measurement) for an input $\vec{x} = (x_1, x_2, \dots, x_L)$
in the qMPS-classifiers.
For any two input sequences $(0,\,x_{2:L})$ and $(1,\,x_{2:L})$ that differ only in their first qubit, 
let
\begin{equation}\label{eq:delta_rho_original}
  \Delta \rho_L\bigl(x_{2:L}\bigr)
  \;=\;
  \rho_L\bigl(0,\,x_{2:L}\bigr)
  \;-\;
  \rho_L\bigl(1,\,x_{2:L}\bigr).
\end{equation}
where $x_{2:L} = (x_2, \dots, x_L)$.
Then, with high probability over the Haar-random draws of $\bigl\{U_k\bigr\}_{k=1}^{L-1}$, the following bounds hold:
\begin{align}
  \Pr\Bigl(\|\Delta \rho_L(x_{2:L})\|_{F} < \epsilon\Bigr)
  &= \frac{1}{\epsilon^2}\,\mathcal{O}\!\bigl({(\tfrac{6}{15})}^{L}\bigr),
  \label{eq:delta_rho_bound} \\[6pt]
  \Pr\Bigl(\Bigl\|\frac{\partial \,\Delta \rho_L(x_{2:L})}{\partial (U_k)_{\alpha\beta}}\Bigr\|_{F} < \epsilon\Bigr)
  &= \frac{1}{\epsilon^2}\,\mathcal{O}\!\bigl({(\tfrac{6}{15})}^{L}\bigr),
  \label{eq:delta_rho_grad_bound}
\end{align}
where \(\|\cdot\|_{F}\) is the Frobenius norm and 
\(\tfrac{\partial}{\partial (U_k)_{\alpha\beta}}\) denotes the partial derivative with respect 
to the \((\alpha,\beta)\) element of unitary matrix \(U_k\).
\end{theorem}
To train on the first-qubit trigger dataset, one must update the parameters so that 
\(\rho_L(0, x_{2:L})\) and \(\rho_L(1, x_{2:L})\) become distinguishable for all 
\(x_{2:L}\). 
\cref{thm:exp_decay} implies that regardless of the choice of loss function,
attempting to push \(\rho_L(0, x_{2:L})\) and \(\rho_L(1, x_{2:L})\) apart through 
gradient-based methods becomes exponentially hard as \(L\) grows.
Notably, this is the first time to prove that there exists a dataset that 
is exponentially hard to train qMPS-circuits for classification task 
from scratch with random initialization. 
You can find more details in the Supplemental Material.

By contrast, classical MPS-classifiers do not suffer from barren plateaus for the same task. 
\begin{theorem}[\textbf{Absence of Barren Plateaus in Classical MPS with Stacked Identity Initialization}]\label{thm:mps_training}
    Consider classical MPS-classifiers,  with tensors $A^{s_k}_{i_k, i_{k+1}}$ 
    initialized using stacked identity with additional small noise in each element\cite{Mossi2024TNInitialization,Dilip2022}.
    Let $\mathcal{L}_{\text{MPS}}(D_{fq})$ denote the loss function calculated by 
    applying softmax function to the output of the MPS-classifiers followed by the negative log-likelihood loss,
    where $D_{fq}$ is the first-qubit trigger dataset.
    Then, the norm of the gradient of $\mathcal{L}_{\text{MPS}}(D_{fq})$ with respect to tensors $A^{s_k}_{i_k, i_{k+1}}$ does not scale with $L$.
    \begin{equation}\label{eq:mps-grad-bound}
    \Bigl\lVert 
    \frac{\partial \,\mathcal{L}_{\text{MPS}}(D_{fq})}{\partial \,A^{s_k}_{i_k, i_{k+1}}}
    \Bigr\rVert
    = O(1).
    \end{equation}
    Consequently, the MPS-classifiers \emph{do not} experience barren plateaus 
    when trained on the first-qubit trigger dataset.
\end{theorem}
While the absence of barren plateaus means gradients do not vanish, it does not strictly guarantee 
that the MPS-classifiers will locate globally optimal parameters; a local minimum could still occur. 
In practice, however, we find empirically (see Supplemental Material) that gradient-based training 
of MPS-classifiers with stacked identity initialization \emph{does} converge to perfect classifiers
on the first-qubit trigger dataset. 
Nonetheless, even if we have the classically trained MPS-classifiers solution, embedding them into quantum circuits
requires a postselection whose success rate falls off \emph{exponentially} with $L$.

\begin{theorem}[\textbf{Exponential postselection Cost for Embedding MPS-Classifiers into qMPS-Classifiers}]\label{thm:post_sel} 
Consider MPS-classifiers that exactly 
classifies the first-qubit trigger dataset and 
we perfectly embed the classifiers into qMPS-classifiers with postselection.
Then, for any input $\vec{x} = (x_1, x_2, \dots, x_L) \in {\{0, 1\}}^L$, the probability $P_{\text{post-sel}}$ where
all intermediate $(L-1)$ qubits are measured in the $\ket{0}$ state, is exponentially small in $L$. 
Concretely, 
\begin{equation} 
P_{\text{post-sel}} = 2^{-(L-1)}.
\end{equation}
\end{theorem}
These theorems show that for the first-qubit trigger dataset, 
although there exist the optimal parameters for qMPS-classifiers,
we've faced two distinct exponential barriers when we try to obtain the qMPS-circuit classifiers with gradient based optimization: 
(i)~attempting to \emph{train} qMPS-classifiers from scratch via gradient descent is obstructed 
by barren plateaus, and (ii)~even though MPS-classifiers itself do not experience barren plateaus,
attempting to use it as a QNN demands exponentially many trials due to postselection.
In the following sections, we show how both challenges can be overcome using an adiabatic 
encoding framework.

\textit{Framework.---}
\begin{figure}[tbp]
	\centering
	\includegraphics[width=0.48\textwidth]{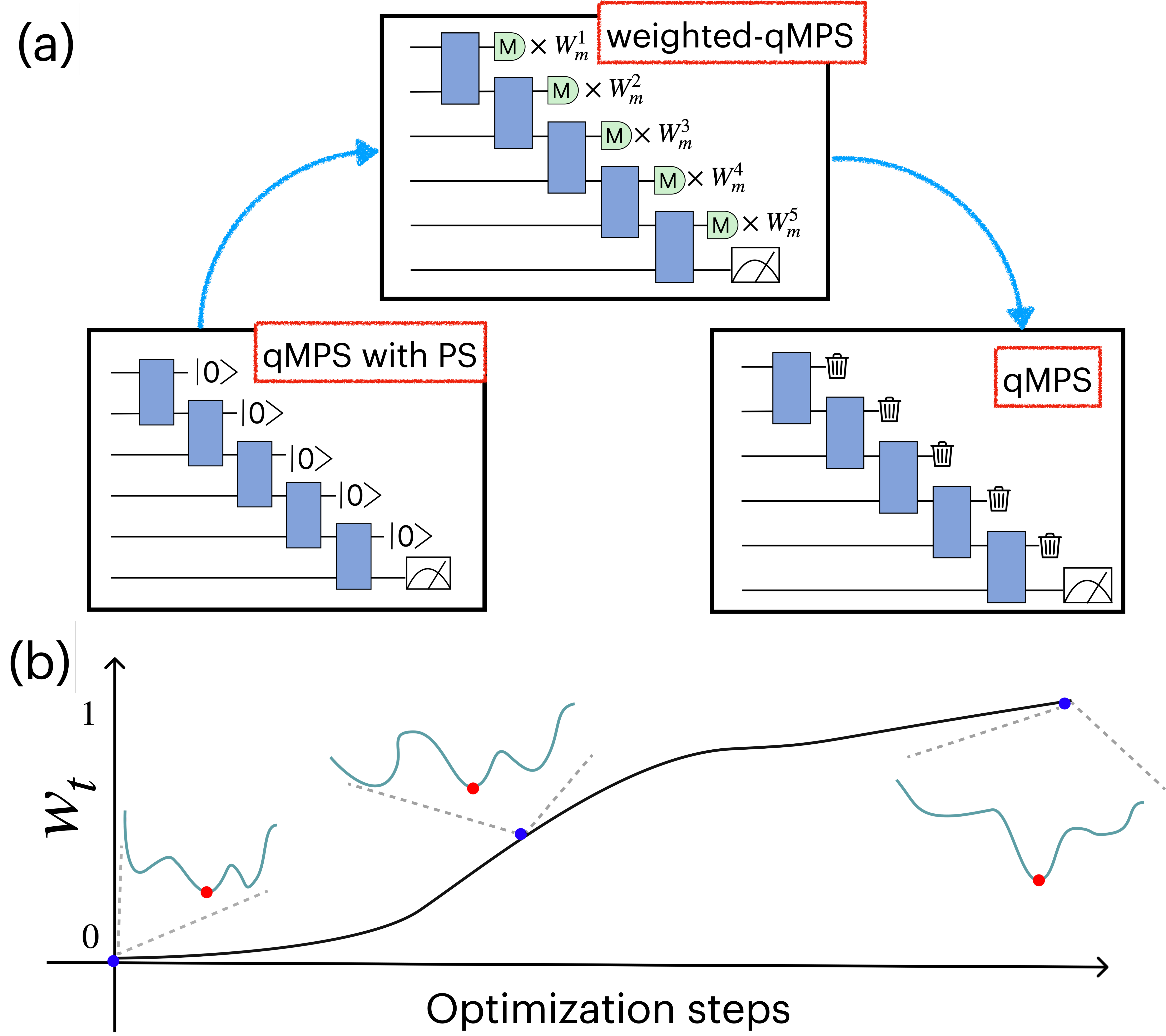}
    \caption{(a) Illustration of a qMPS-classifier with postselection (left), 
    a qMPS-classifier without postselection (right), and a weighted qMPS ansatz 
    (middle). The green boxes labeled $M$ represent measurements on each qubit.
    (b) Schematic of how the loss  landscape changes as $W_1$ is gradually increased from 0 to 1.}\label{fig:ae}
\end{figure}
Having observed that a naive embedding of the MPS-classifiers into quantum circuits
inevitably requires postselection, we now introduce a framework that circumvents 
the severe overhead. 
A key element of our construction is the notion of \emph{weighted quantum states} first introduced in~\cite{Holmes2023}. 
In its most general form, we can view weighted quantum states as circuits whose measurement 
outcomes are assigned classical weights before combining them into a single effective state. 
Concretely, consider a general quantum instrument $\{\mathcal{E}_{m}\}$ acting on an input 
density matrix $\rho_{\text{in}}$, with corresponding outcomes labeled by $m$, 
then the \emph{weighted quantum states} is defined as
\begin{equation}
	\tilde{\tau}
	= 
	\sum_{m} W_m \, p_m \, \frac{\mathcal{E}_m\bigl(\rho_{\text{in}}\bigr)}{\mathrm{Tr}\!\Bigl[\mathcal{E}_m\bigl(\rho_{\text{in}}\bigr)\Bigr]}
	=
	\sum_{m} W_m \,\mathcal{E}_m\bigl(\rho_{\text{in}}\bigr),
\end{equation}
where $p_m = \mathrm{Tr}[\mathcal{E}_m\bigl(\rho_{\text{in}}\bigr)]$ 
is the probability of obtaining outcome $m$, 
and $W_m$ is the weight assigned to that outcome. 
Notice that in general $\mathrm{Tr}[\tilde{\tau}] \neq 1$, 
so $\tilde{\tau}$ is not a valid density operator. 
Nevertheless, when used in classifiers, 
the overall scalar factor does not affect the classification result.
Applying this concept to the qMPS-classifiers leads to an ansatz illustrated in~\cref{fig:ae} (a)
which we call \emph{weighted qMPS}.
In the weighted qMPS, the final measurement outcome is obtained by assigning a weight $W_m^{j}$ 
depending on the measurement outcome $m \in \{0, 1\}$ of the $j$-th qubit.
Despite the significant differences between qMPS-classifiers with postselection and qMPS-classifiers without postselection, 
this weighted qMPS interpolates between these two extremes.
For \emph{qMPS-classifiers with postselection} set $W^j_m = (1, 0)$ 
so that the circuits keep running only when the measurement outcome is $0$, 
and for \emph{qMPS-classifiers without postselection} set $W^j_m = (1, 1)$ to realize partial trace.

To achieve a smooth and gradual transition from 
$W_m^{j} = (1, 0)$ to $W_m^{j} = (1, 1)$, 
it is necessary to determine an appropriate scheduling for $w$ in $W_m = (1, w)$,
where we assume to update $W_m^{j}$ uniformly for all $j$. 
Specifically, given a total optimization iteration count $T$, we predefine a schedule $w_t$ for $t = 0, 1, \dots, T$, 
so that $w_t$ monotonically increases from $w_0 = 0$ to $w_T = 1$.
~\cref{fig:ae} (b) shows the concept of this process.
The loss function landscape also deforms gradually with the schedule $w_t$.
At each $w_t$, we optimize the unitary gates until the loss function becomes small enough.
This scheduling process involves engineering considerations, 
as we have observed that different datasets require different pacing. 
Additionally, we have found that introducing a regularization term 
$
\mathcal{R}(W) = -\log \mathrm{Tr}[\tilde{\tau}],
$
which can be interpreted as a measure of success rate $P_{\text{success}}$, 
further stabilizes the training process. 

\textit{Numerical Experiments.---}
Before demonstrating how our adiabatic encoding framework resolves the barren plateaus on two representative datasets, 
we briefly outline our implementation settings.
For the MPS-classifiers, we used the Adam optimizer.
For the qMPS-classifiers, whose unitary gates must be updated, we employed a ``Riemannian Adaptive Optimization Method'', 
which effectively implements an Adam-like algorithm on Riemannian manifolds~\cite{Becigneul2018}. 
Both classifiers utilized the same feature map $\phi(x)\propto (x,\,1-x)$.
Further implementation details are provided in the Supplemental Material.

\begin{figure}[h]
	\centering
	\includegraphics[width=0.48\textwidth]{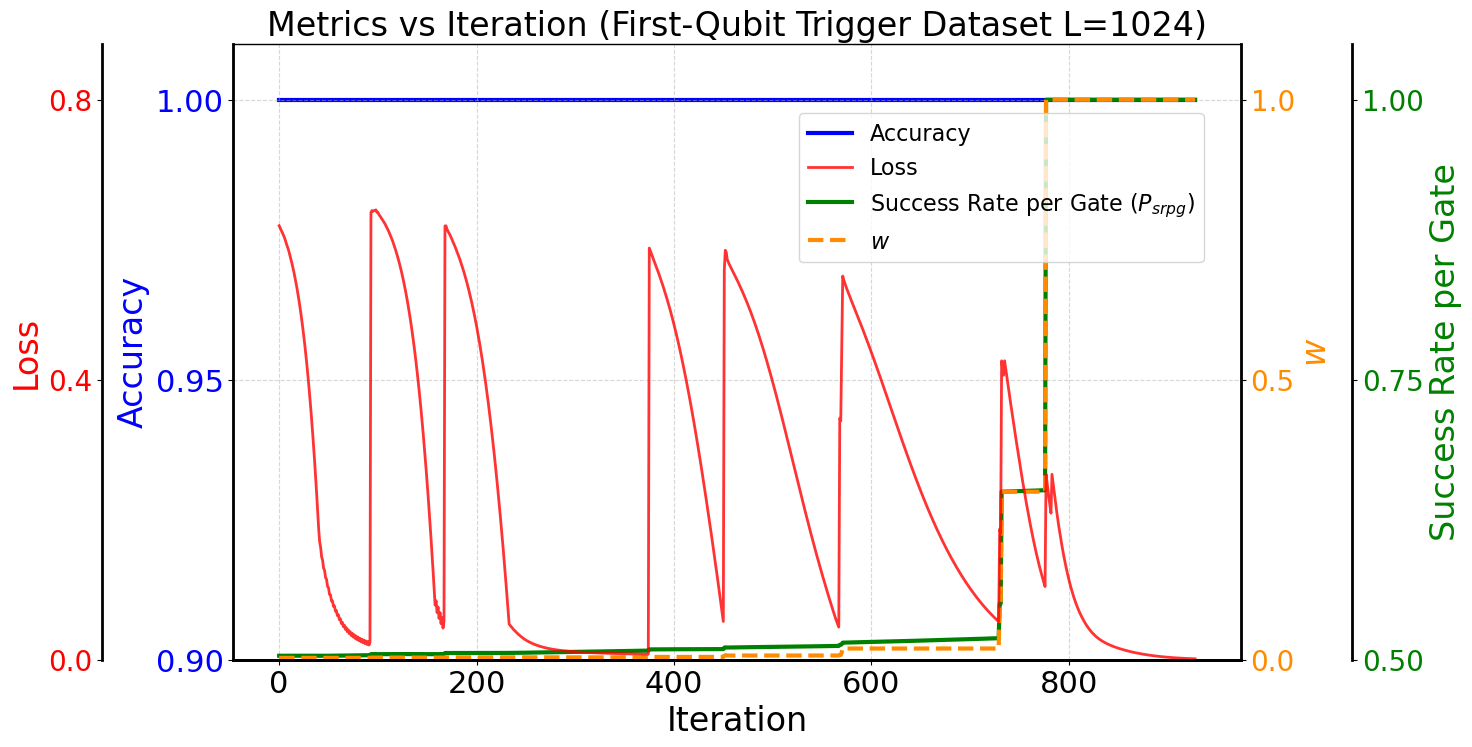}
	\caption{Adiabatic encoding framework for the first-qubit trigger dataset.
	Here, ``success rate per gate'' denoted by $P_{\text{srpg}}$ is the average success probability at a each two-qubit gate,
    which can be calculated from $\mathrm{Tr}[\tilde{\tau}]$ at each $k$-th gate.
	The overall circuit's success probability $P_{\text{success}}$ can be calculated by $P_{\text{success}} = P_{\text{srpg}}^{{L-1}}$.
    }
	\label{fig:ae_first_qubit_trigger}
\end{figure}

We now demonstrate these methods on the first-qubit trigger dataset, 
with $L = 1024$.
Directly training a qMPS-classifier on this dataset 
is unfeasible due to the severe barren plateaus.
By contrast, we can 
prepare a qMPS-classifier with postselection that perfectly classifies the dataset by exactly embedding a 
pre-trained MPS-classifier 
then use adiabatic encoding to remove postselection.
Throughout this procedure, the classification accuracy remains at 100\%, as shown in 
\cref{fig:ae_first_qubit_trigger}. We train on mini-batches of $2^{12}$ examples 
per iteration. 
In every update of $w$, the loss function increases discontinuously but we train the circuit until the loss function 
becomes small enough.
Empirically, we find that, from the initial $w_1 = 0.002$, the $w$ must be 
increased very gradually in the early phase of adiabatic encoding to 
maintain stability. However we also observed that once $w$ reaches approximately 0.02, faster 
increments do not degrade performance, thus accelerating the scheduling pace.

\onecolumngrid 
\begin{center}
  \includegraphics[width=0.95\textwidth]{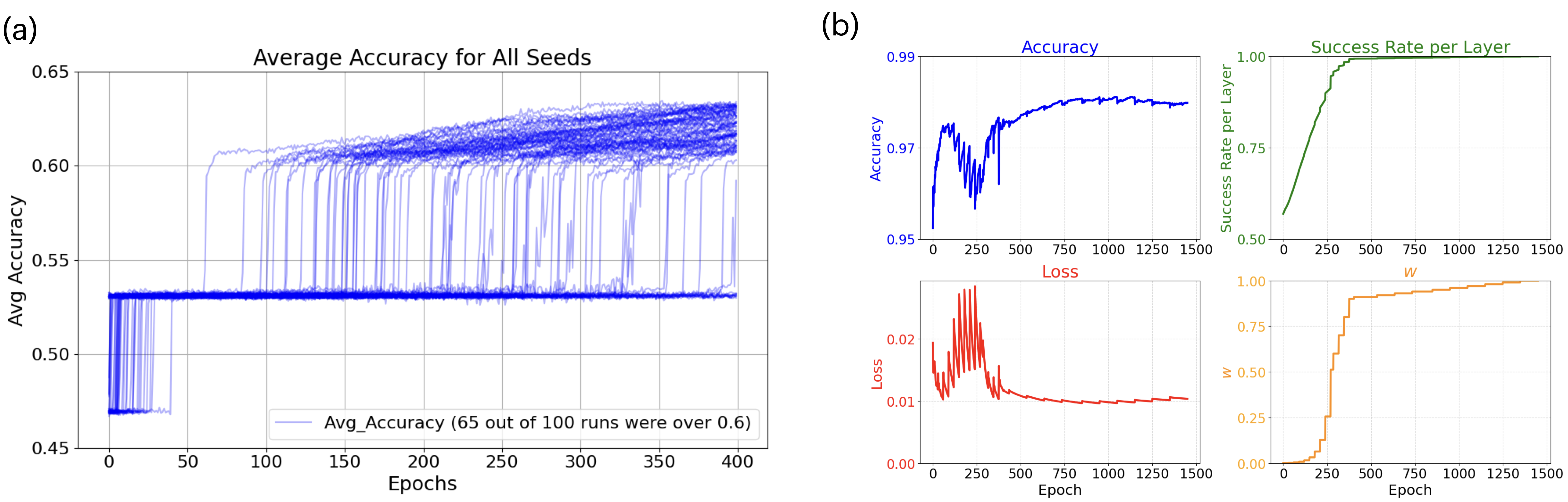}
  \captionof{figure}{(a) Direct training of the qMPS-classifier suffers from barren plateaus and local minima. (b) Adiabatic encoding successfully circumvents these issues. 
  In the same manner, $P_{\text{success}}$ can be calculated by $P_{\text{success}} = P_{\text{srpg}}^{{256 - 1}}$.
  Both (a) and (b) are for training dataset.}\label{fig:ae_binary_mnist}
\end{center}
\twocolumngrid
To assess the performance on a more realistic dataset, we next consider the binary MNIST dataset of digits 3 and 5, 
which we have found to be the most difficult case to distinguish with a qMPS-classifier.
We can attribute this difficulty to the fact that 3 and 5 are very similar from middle to bottom.
The dataset contains 11552 training samples, where each sample is a $16 \times 16$ pixel image reshaped into a 256-dimensional vector.
\cref{fig:ae_binary_mnist} (a) shows that training the qMPS-classifier directly suffers from barren plateaus and local minima.
We've conducted 100 independent trials with different random seeds but all of runs failed due to either barren plateaus or local minima:
only 35\% of runs reached above 60\% accuracy and none exceeded $70\%$ accuracy.
However, by pre-training the MPS-classifier for 100 epochs and applying adiabatic encoding, 
we successfully circumvent these issues. 
The results in~\cref{fig:ae_binary_mnist} (b) indicate that our method can convert the pre-trained 
MPS-classifier which requires postselection into a qMPS-classifier without postselection
while maintaining high accuracy.
Similar to the previous case, we observe temporary spikes in the loss at each update from $w_t$ to $w_{t+1}$ but we train the circuit until the loss function 
becomes small enough to proceed to the next $w$.
You can suppress the spikes by employing a more gradual scheduling strategy but it is not necessary for the final performance.
The reason for the slight performance improvement is believed to be that while the MPS-classifier was only trained for 100 epochs, the adiabatic encoding process involved exposure to the data for over 1400 epochs.

\textit{Conclusion.---}
In this work, we have introduced an adiabatic encoding framework 
that converts classically trained MPS-classifiers into qMPS-classifiers without degrading the performance.
To demonstrate the effectiveness of our method, we applied it to two datasets: 
the first-qubit trigger dataset and the binary MNIST dataset.
The first-qubit trigger dataset is an artificial dataset such that the training of qMPS-classifiers is mathematically proven to be exponentially hard due to barren plateaus.
We also showed that while training of MPS-classifiers on this dataset itself is not affected by barren plateaus,
the embedding them into quantum circuits requires exponentially many trials for postselection.
The binary MNIST dataset is a real dataset that is challenging to directly train qMPS-classifiers because of the barren plateaus and local minima.
Despite both datasets posing significant training challenges, we have shown that our adiabatic encoding framework ensures that both can be trained successfully.

We mainly focused on the adiabatic encoding of the MPS-classifiers with $\chi=2$ but applying it to other tensor network
structures and larger $\chi$ is also possible.
Overall, our framework provides a scalable blueprint for 
exploiting the synergy between tensor-network and quantum neural networks
and helps accelerate progress toward genuine quantum advantage in machine learning.

\begin{acknowledgments}
The authors would like to thank Tsuyoshi Okubo, 
Synge Todo, Takumi Kobori
and Shota Sugawara for discussions and comments.
This work was supported by the Center of Innovation for Sustainable Quantum AI, JST Grant Number JPMJPF2221, and JSPS KAKENHI
Grants No. JP24KJ0892.
\end{acknowledgments}

\bibliographystyle{apsrev4-2} 
\bibliography{main}

%
%
%



\onecolumngrid

\section*{Supplemental Material for "Adiabatic Encoding of Pre-trained MPS Classifiers into Quantum Circuits"}

\subsection{Implementation Details}

\subsubsection{Classical MPS-Classifiers}\label{sssec:supp_mps_impl}
We train the MPS-classifiers by the Adam optimizer. 
The classification output is obtained by applying a softmax to $f_\ell(\vec{x})$ as 
\begin{equation*}
  \label{eq:softmax}
  \sigma(\vec{x})_{\ell} = \frac{e^{f_\ell(\vec{x})}}{\sum_{\ell'} e^{f_{\ell'}(\vec{x})}}, \tag{S1}
\end{equation*}
and the loss function is the negative log-likelihood (NLL). 
\begin{align*}
  \label{eq:nll}
  \mathcal{L}(\mathcal{D}) &= \frac{1}{N_{\mathcal{D}}}\sum_{(\vec{x}_i, \ell_i) \in \mathcal{D}} \mathcal{L}(\vec{x}_i, \ell_i) \\
  &= -\frac{1}{N_{\mathcal{D}}}\sum_{(\vec{x}_i, \ell_i) \in \mathcal{D}} \log(\sigma(\vec{x}_i)_{\ell_i}). \tag{S2}
\end{align*}
We initialize the MPS tensors with “stacked identity plus noise”.
More specifically, 3-rank tensors $A^{s_k}_{i_k, i_{k+1}}$
are initialized as 
\begin{equation*}\label{eq:mps_init}
  A^{s_k}_{i_k, i_{k+1}} = \delta_{i_k, i_{k+1}} + \mathcal{N}^{s_k}_{i_k, i_{k+1}}, \tag{S3}
\end{equation*}
where $\mathcal{N}^{s_k}_{i_k, i_{k+1}}$ is a random variable drawn from the normal distribution with mean 0 and variance $\epsilon^2$.
We select $\epsilon$ so that $\sqrt{L}\,\epsilon \ll 1$.
In the numerical experiments, we set $\epsilon = 0.1/\sqrt{L}$.
For the leftmost matrix, any initialization that does not depend on $s_{1}$ is permissible in principle. 
Nonetheless, a common convention is to set $A_{i_1}^{s_1} = 1 + N_{i_1}^{s_1}$. 
One may view this as treating the first tensor 
as a rank-3 array with a common initialization $A_{i_0, i_1}^{s_1} = \delta_{i_0,i_1} + N_{i_0,i_1}^{s_1}$, 
followed by imposing the boundary condition $\langle + \vert$ on the left side.
The core idea of this 
“stacked identity” initialization is that each MPS layer produces 
an output that is not biased by the input data.
Concretely, if the $k$-th tensors were 
purely identity (i.e., without random noise), then for all input states the MPS 
output vector would simply be ${(1,\,1)}^T$ at every layer, which predicts 50\% for both classes. 
Introducing random perturbations around 
this identity baseline ensures that the model's outputs are slightly different 
for different input states, which is crucial for the model to learn.

\subsubsection{qMPS-Classifiers via Adiabatic Encoding}
\label{sec:supp_qmps_impl}
As explained in the main text, we perform an adiabatical shift from weights 
$W_m^{j}=(1,0)$ (postselection) to $W_m^{j}=(1,1)$ (no postselection). 
Generally, we can set independet scheduling for each weight $W_m^{j}$.
However, in this work, we use uniform scheduling for all weights $W_m^{j} = W_m = (1, w)$.
A schedule for $w_t$ for $t=0,1,\ldots,T$ is chosen so that 
we discretize the interval $[0,1]$ in $T$ steps and gradually increase $w_t$ during training, 
re-optimizing parameters at each step. 
As discussed in the main text, the scheduling process involves engineering considerations,
and often require a dataset-dependent tuning by the user.
During the adiabatic encoding, we train the parameters until the loss converges at each step of $w_t$.
The optimization is done with Riemannian Adaptive Optimization Method~\cite{Becigneul2018} 
as we need to optimize parameters on unitary group ($U(4)$), a type of Riemannian manifold.
Unlike the MPS-classifiers whose output $f_\ell(\vec{x})$ is not normalized,
the qMPS-classifiers naturally returns probability distribution over the classes
and we can directly use the output for NLL to calculate the loss.
Also, we have observed that in some cases, replacing the NLL with focal loss~\cite{Lin2018}
can accelerate the scheduling pace for $w_t$, and we indeed use focal loss for the MNIST dataset.

In the main text, we have used $\Tr[\tilde{\tau}]$ as a regularization term,
which can be interpreted as the success probability of postselection.
In more detail, consider the measurement on the $j$-th qubit, which can yield outcome $m\in\{0,1\}$. 
One may view the weighting $W_m^j$ as multiplying the final output by $W_m^j$, 
yet for many classification tasks we only need to compare the relative scale of the outputs. 
This means the scailing factor does not affect the classification results.
This observation enables a “pseudo-postselection” scheme, wherein you can probabilistically conduct postselection.

Concretely, suppose we set $W_m = (1, w)$ (with $0 \le w \le 1$) uniformly for all qubits. 
Then for an input state $\rho$,
\begin{equation*}\label{eq:quasi_ps}
    \tilde{\tau} \;=\; \mathcal{E}_0(\rho) \;+\; w\,\mathcal{E}_1(\rho), \tag{S4}
\end{equation*}
where $\mathcal{E}_0$ (resp.\ $\mathcal{E}_1$) denotes the completely positive (CP) map obtained when the measurement outcome is $0$ (resp.\ $1$). 
Notice that if $w=1$, the overall map becomes trace-preserving (no genuine postselection), 
while for $w=0$, the quantity $\mathrm{Tr}[\tilde{\tau}]$ equals 
$\mathrm{Tr}[\mathcal{E}_0(\rho)]$, precisely the success probability of standard postselection onto outcome $m=0$. 
Hence for intermediate values of $w$, \(\tilde{\tau}\) can be interpreted as continuing the circuits with probability $w$ even after measuring outcome $m=1$, 
making $\mathrm{Tr}[\tilde{\tau}]$ a success probability for quasi-postselection procedure. 
In this sense,~\cref{eq:quasi_ps} gives a unified way to see how $\mathrm{Tr}[\tilde{\tau}]$ tracks 
the effective fraction of “kept” measurement outcomes as $w$ is varied between postselection ($0$) and trace-preserving ($1$).

\subsubsection{Feature Map}
\label{sec:supp_featmap}
In both the classical MPS and qMPS-classifiers, we use the same feature map 
\(\phi(x)\propto (x,1-x)\) and each input element $x$ is preprocessed to lie in $[0,1]$ whenever necessary. 
In MPS-classifiers, the resulting vectors are normalized so that $\sum_i \phi_i(x)=1$.
This normalization is essential for the MPS-classifiers to avoid the explosion or vanishing of the norm.
In qMPS-classifiers, the feature map is applied to each qubit, and the resulting vectors are normalized so that $\sum_i ||\phi_i(x)||^2 = 1$,
which is a fundamental requirement for the input state to be a valid quantum state.
This scailing does not affect the classification results as only relative size relation is important.
We can also use other feature maps, e.g., $\psi(x) \propto (\cos(x), \sin(x))$,
which is also a well-used feature map for TN-classifiers and qTN-classifiers.

In the following, we provide detailed proofs of the theorems regarding the trainability of qMPS-classifiers and MPS-classifiers for the First-Qubit Trigger dataset.

\subsection{Proof of Theorem 1: Exponential Indistinguishability of qMPS-Classifiers}

\noindent \textbf{First and Second Moments of Haar-Random Unitaries.} 

We begin by recalling the first and second moments of Haar-random unitaries, 
which are essential to the proof of Theorem~1. 
For a unitary matrix $U$ in $U(N)$ with Haar measure $dU_H$, 
where $N$ is the dimension of the unitary matrix, 
the first moment $M_1(dU_H)$ and second moment $M_2(dU_H)$ satisfy the following integrals: 
\begin{align}
    M_{1}(dU_H) &= \int_{U(N)} dU_H \, U_{l_0,r_0} \, \overline{U}_{l_0',r_0'} = \frac{1}{N} \delta_{l_0,l_0'} \delta_{r_0,r_0'}, \tag{S5} \\[6pt]
    M_{2}(dU_H) &= \int_{U(N)} dU_H \, U_{l_0,r_0} \, U_{l_1,r_1} \, \overline{U}_{l_0',r_0'} \, \overline{U}_{l_1',r_1'} \nonumber \\
    &\quad= \frac{1}{N^2 - 1} \Bigl( \delta_{l_0,l_0'} \delta_{l_1,l_1'} \delta_{r_0,r_0'} \delta_{r_1,r_1'} + \delta_{l_0,l_1'} \delta_{l_1,l_0'} \delta_{r_0,r_1'} \delta_{r_1,r_0'} \Bigr) \tag{S6} \\[6pt]
    &\quad\quad - \frac{1}{N(N^2-1)} \Bigl( \delta_{l_0,l_0'} \delta_{r_0,r_0'} \delta_{l_1,l_1'} \delta_{r_1,r_1'} + \delta_{l_0,l_1'} \delta_{r_0,r_1'} \delta_{l_1,l_0'} \delta_{r_1,r_0'} \Bigr), \tag{S7}
\end{align}
where $\delta_{a,b}$ is the Kronecker delta and 
$\overline{U}$ denotes the complex conjugate of each element of $U$. 
As discussed in Ref.~\cite{Liu2022}, the above properties can also be represented 
using tensor network diagrams, which provide a convenient visual and 
computational framework for handling such integrals over the unitary group

\begin{equation}
  \label{eq:m1_tn_o}
  \begin{array}{c}
  \includegraphics[width=0.3\textwidth]{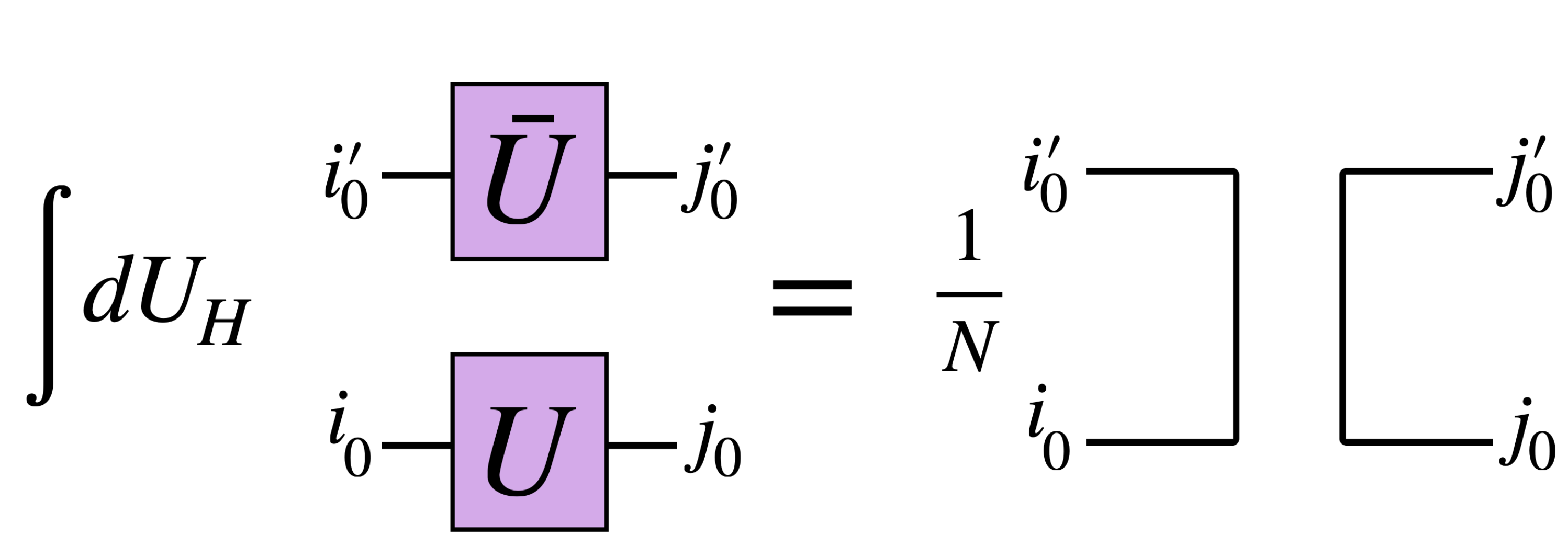}\hspace{3pt}, \tag{S8} \\[6pt]
  \end{array}
\end{equation}

\begin{equation}
  \label{eq:m2_tn_o}
  \begin{array}{c}
  \includegraphics[width=0.95\textwidth]{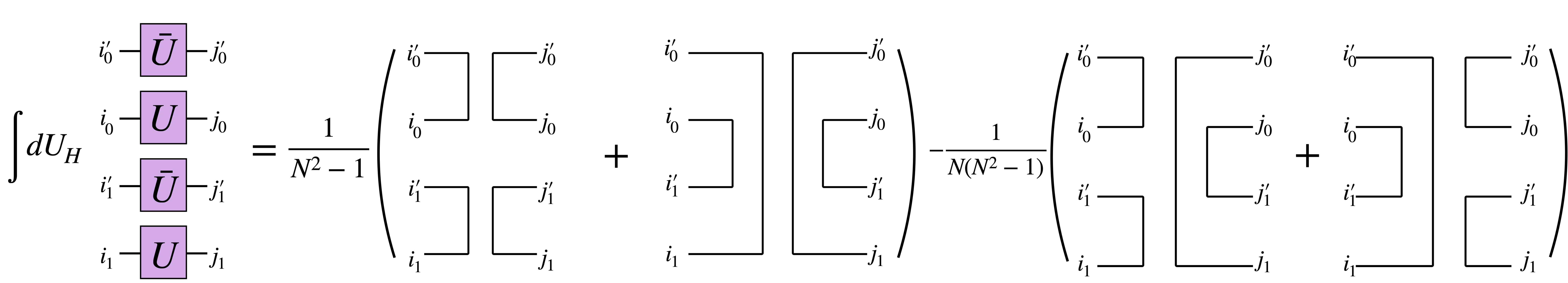}\hspace{3pt}. \tag{S9} \\[6pt]
  \end{array}
\end{equation}

In the case of two-qubit unitary gates, this tensor network diagram can also be
expressed as

\begin{equation}
  \label{eq:m1_tn}
  \begin{array}{c}
  \includegraphics[width=0.35\textwidth]{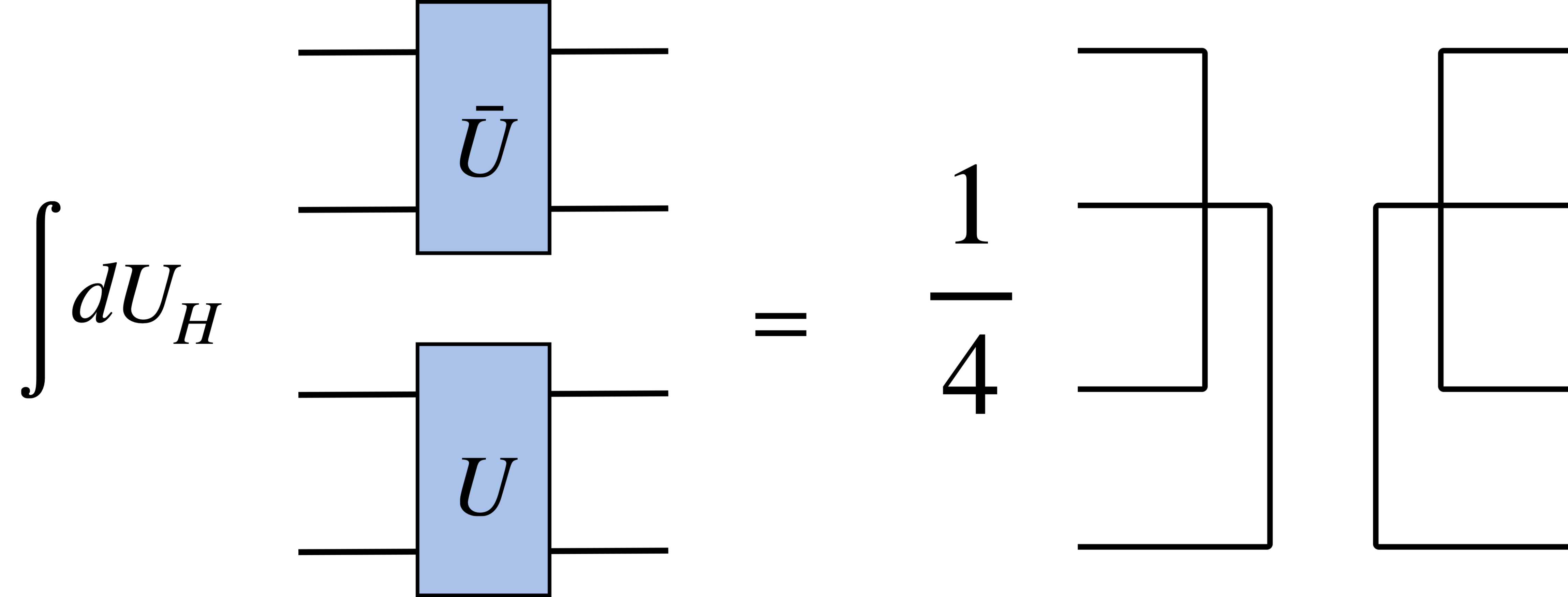}\hspace{3pt}, \tag{S10} \\[6pt]
  \end{array}
\end{equation}

\begin{equation}
  \label{eq:m2_tn}
  \begin{array}{c}
  \includegraphics[width=0.7\textwidth]{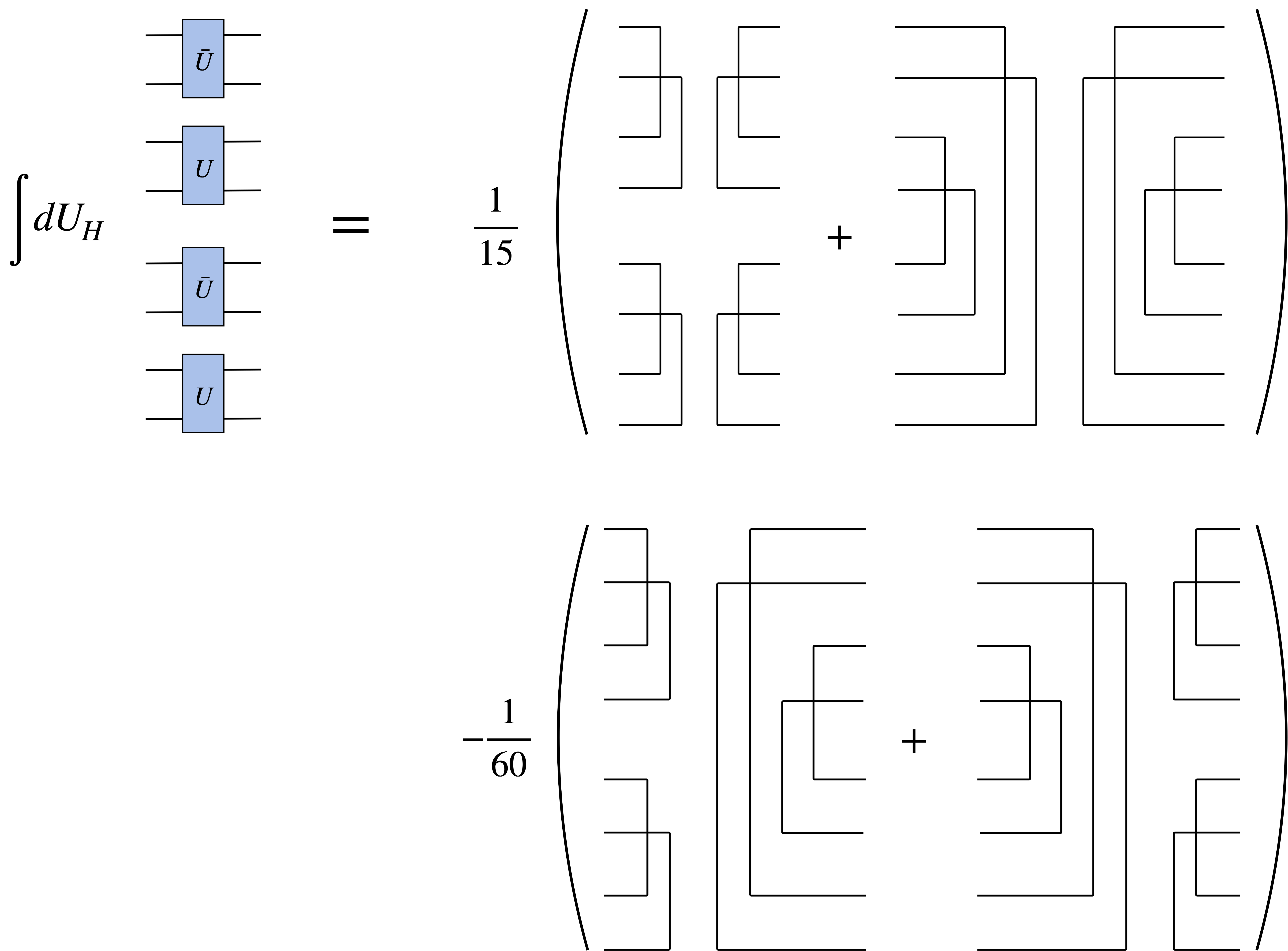}\hspace{3pt}. \tag{S11} \\[6pt]
  \end{array}
\end{equation}

Note that two-legs tensor each has 4 dimensions now expressed as four-legs tensor. each has 2 dimensions
We used $N=4$ since two-qubit unitary gates are in $U(4)$.

\noindent \textbf{Mean and the variance of $\Delta \rho$.}

From now on, we investigate the expectation 
$\mathbb{E}[\Delta \rho]$ and variance $\mathrm{Var}[\Delta \rho]$ under the Haar measure $dU_H$.
Without loss of generality, we set all $x_{2:l}$ to zero. 
Then, $\Delta \rho(x_{2:l})$ can be represented by the tensor network in~\cref{eq:delta_rho}.

\begin{equation}
  \label{eq:delta_rho}
  \begin{array}{c}
  \includegraphics[width=0.6\textwidth]{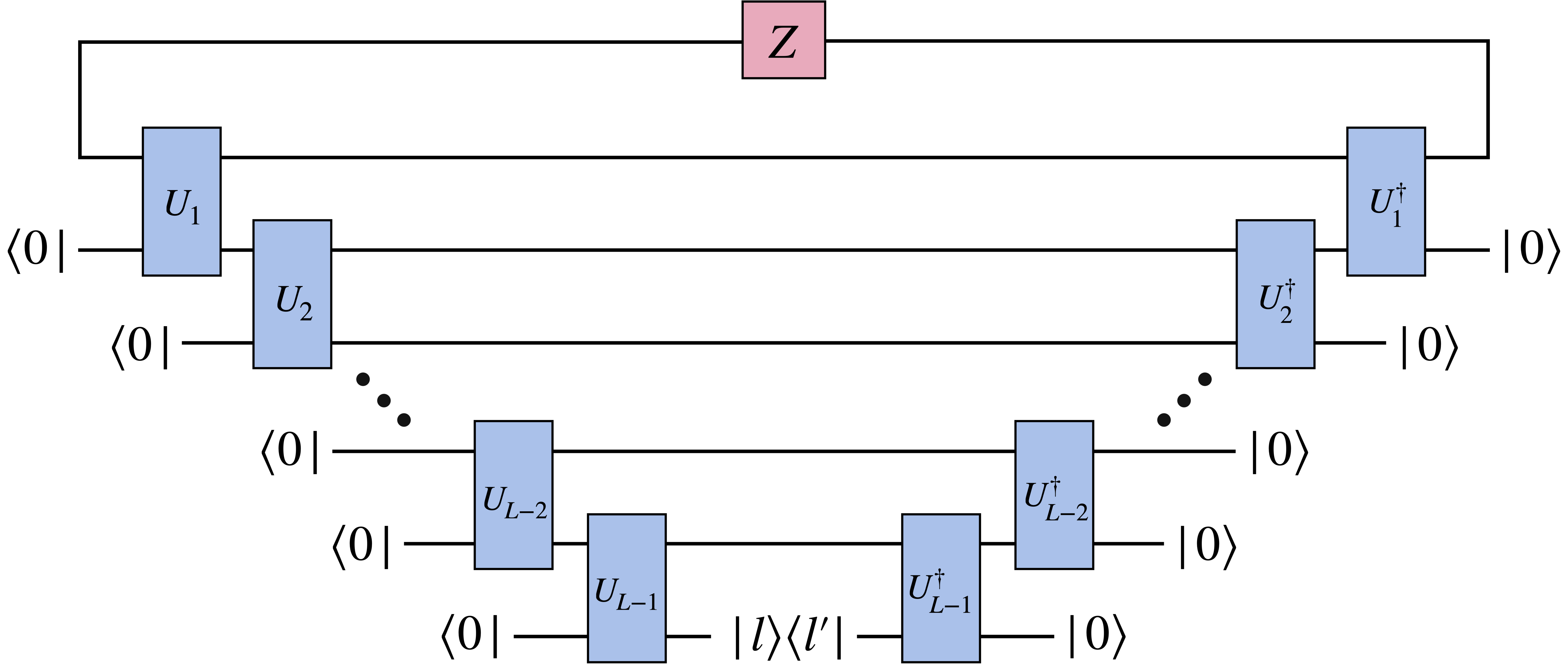}\hspace{3pt}, \tag{S12} \\[6pt]
  \end{array}
\end{equation}
By rearranging the network in~\cref{eq:delta_rho}, we obtain a one-dimensional layout.
\begin{equation}
  \label{eq:delta_rho_linear}
  \begin{array}{c}
  \includegraphics[width=0.5\textwidth]{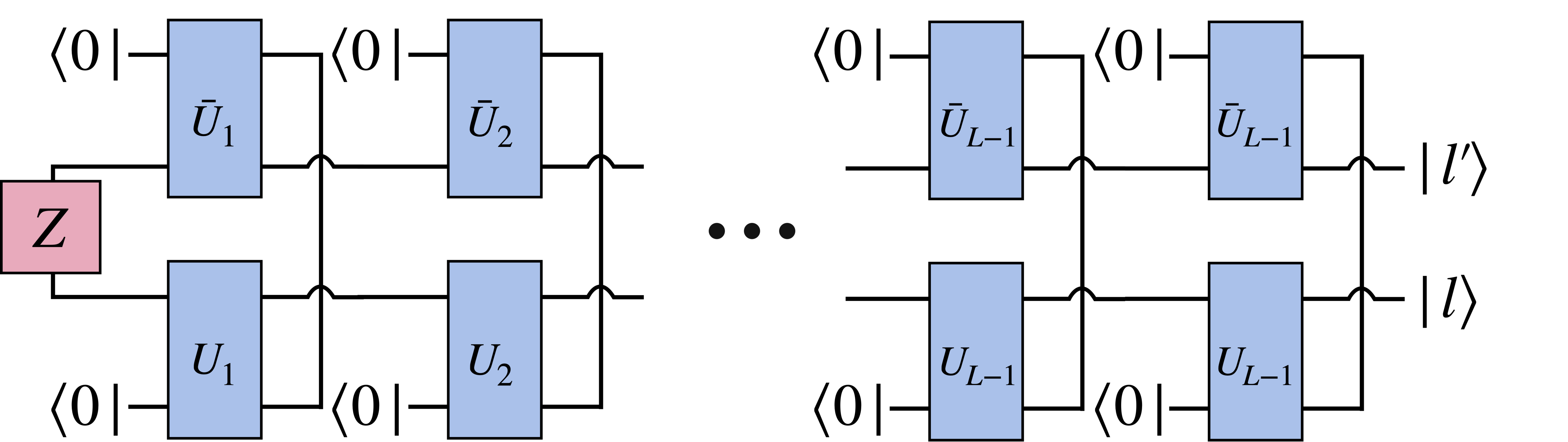}\hspace{3pt}, \tag{S13} \\[6pt]
  \end{array}
\end{equation}
Since we flip the four-leg tensor horizontally, $U^\dagger$ becomes $\bar{U}$.

\begin{lemma}\label{lemma:mean-deltarho}
  Expectation value of $\Delta \rho$ vanishes,
  \begin{equation}
    \mathbb{E}[\Delta \rho] = \mathbf{0}. \tag{S14}
  \end{equation}
  and the variance of $\Delta \rho$ decays exponentially as the system size $L$ increases,
  \begin{equation}
    \mathrm{Var}[\Delta \rho] = O((\frac{6}{15})^{L}). \tag{S15}
  \end{equation}
\end{lemma}

\begin{proof}
First, we show $\mathbb{E}[\Delta \rho] = 0$. Using~\cref{eq:m1_tn} and~\cref{eq:delta_rho_linear}, 
we integrate the $k$-th unitary gate as shown in~\cref{eq:exp_delta_2}.
\begin{equation}
  \label{eq:exp_delta_2}
  \begin{array}{c}
  \includegraphics[width=0.65\textwidth]{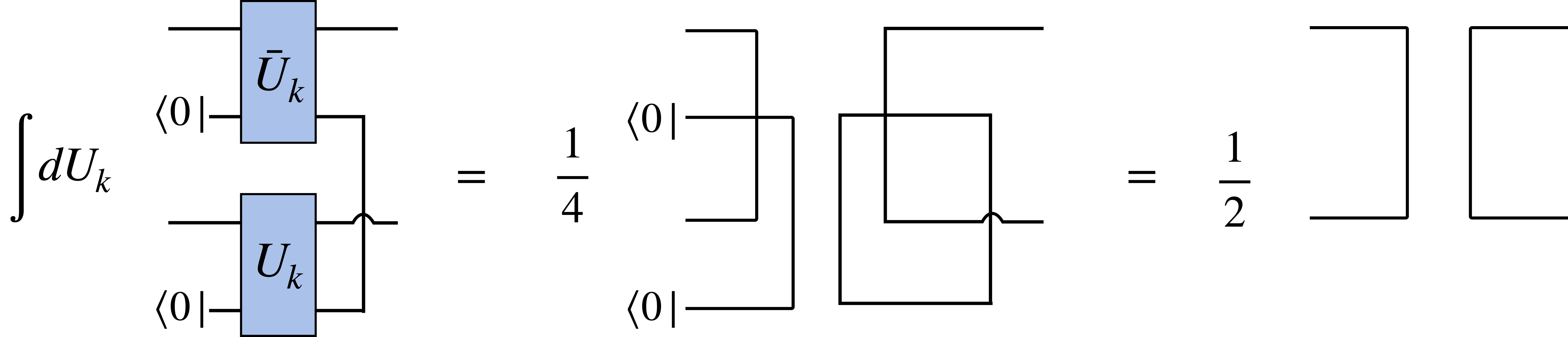}\hspace{3pt}, \tag{S16} \\[6pt]
  \end{array}
\end{equation}
Here we used the fact that the loop in the second term is~\cref{eq:exp_delta_2} simply becomes $2$.
By denoting the result of the integration up to the $(k-1)$-th unitary gate as $\mathrm{Q}_{k-1}$, we can formulate a recursive expression
\begin{equation}
  \label{eq:exp_delta_rec}
  \includegraphics[width=0.62\textwidth]{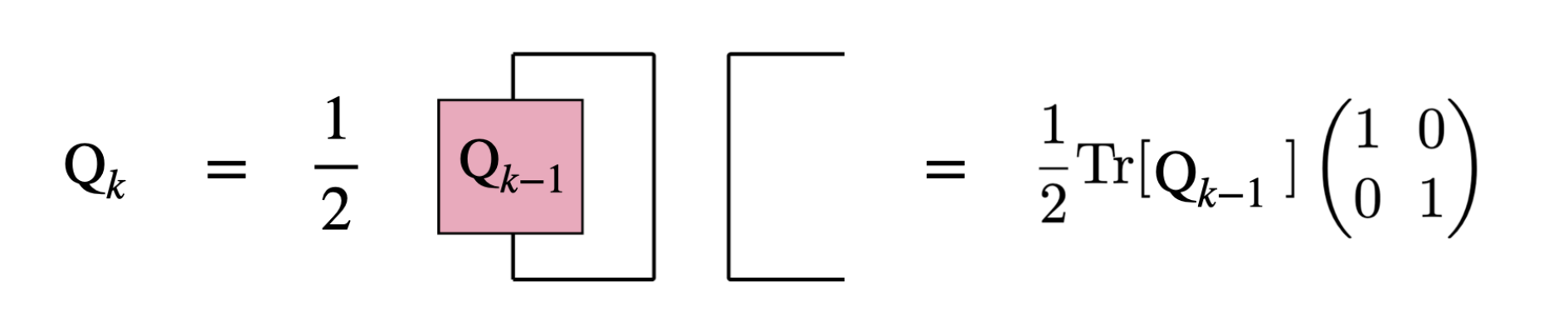}\hspace{3pt}. \tag{S17} \\[6pt]
\end{equation}
Here we expressed the open line with two legs as identity matrix. 
By noting that $\mathrm{Tr}[Z] = 0$, we conclude $\mathbb{E}[Q_L] = \mathbf{0}$.

Next, we show that the variance of $\Delta \rho$ decays exponentially as the system size $L$ increases.
By since $\mathbb{E}[\Delta \rho] = \mathbf{0}$, the variance of $\Delta \rho$ 
simply becomes $\mathbb{E}[{(\Delta \rho)}^{\odot 2}]$, 
where ${(\cdot)}^{\odot 2}$ element-wise square.
The tensor network diagram of $\mathbb{E}[{(\Delta \rho)}^{\odot 2}]$ is shown in~\cref{eq:exp_delta}.

\begin{equation}
  \label{eq:exp_delta}
  \begin{array}{c}
  \includegraphics[width=0.65\textwidth]{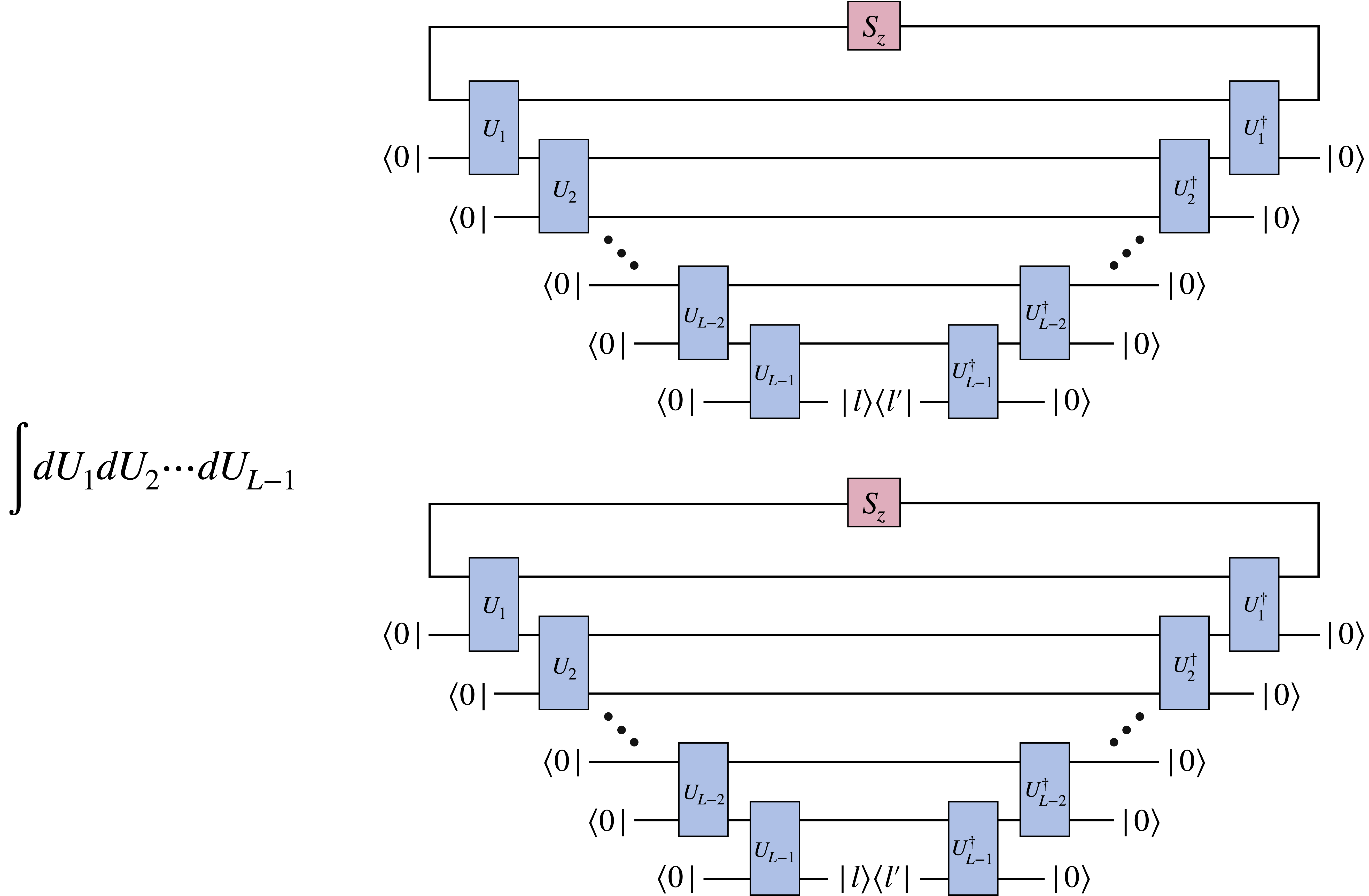}\hspace{3pt}, \tag{S18} \\[6pt]
  \end{array}
\end{equation}

Applying~\cref{eq:m2_tn} to~\cref{eq:exp_delta}, we can calculate the integral of the $k$-th unitary gate as in~\cref{eq:exp_delta_4}.
\begin{equation}
  \label{eq:exp_delta_4}
  \begin{array}{c}
  \includegraphics[width=0.9\textwidth]{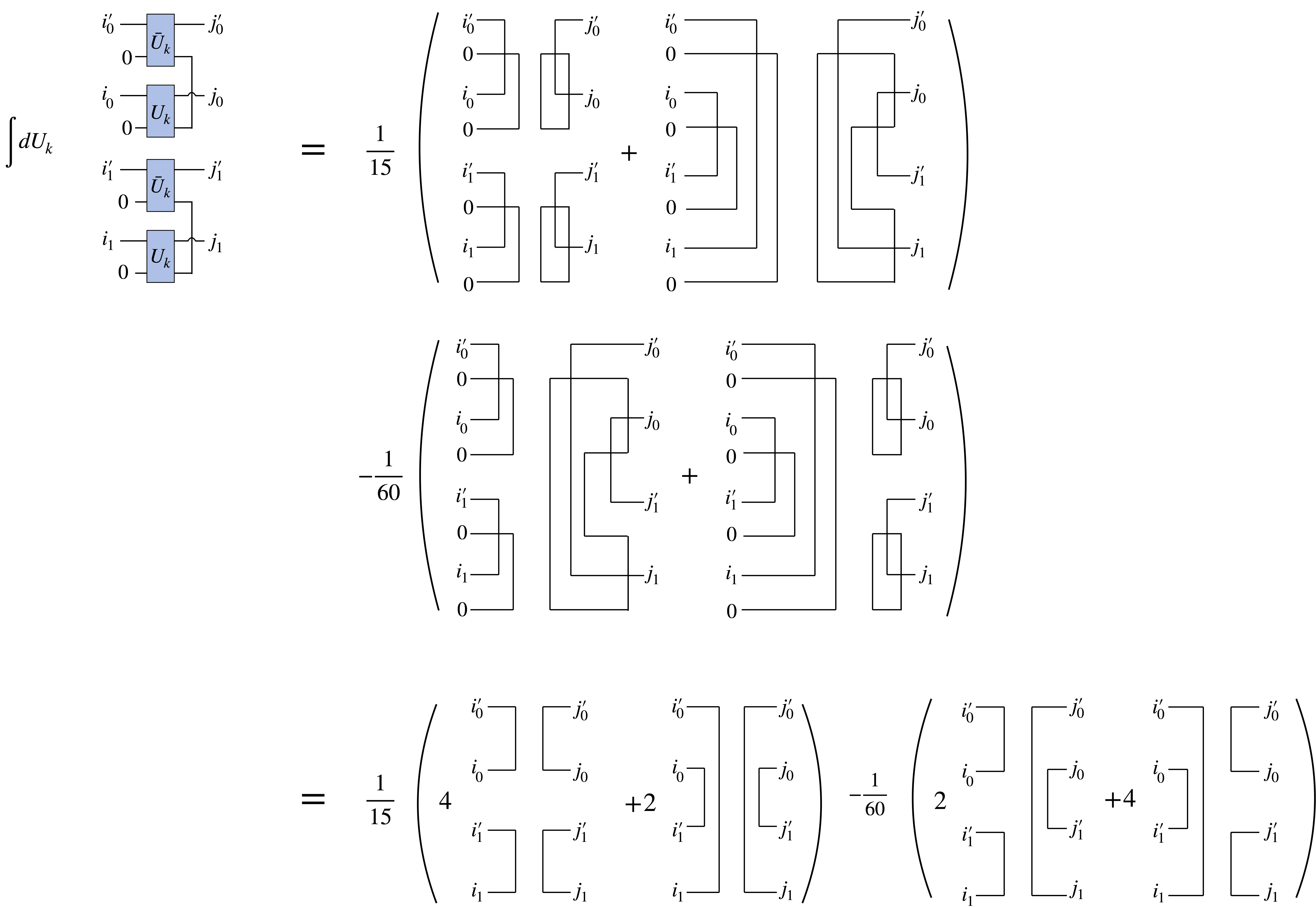}\hspace{3pt}, \tag{S19} \\[6pt]
  \end{array}
\end{equation}

We again denote the result of the integration up to the $(k-1)$-th unitary gate as $Q^{(2)}_{k-1}$ and derived 
\cref{eq:a_k,eq:b_k,eq:a_k_def} as the recursive expressions.

\begin{align}
  A_{k+1} & = \frac{16}{15} A_k - \frac{4}{60}A_k + \frac{4}{15} B_k - \frac{16}{60} B_k \nonumber \\
          & = A_k, \tag{S20} \label{eq:a_k} \\
  B_{k+1} & = \frac{8}{15} B_k - \frac{8}{60}B_k + \frac{8}{15} A_k - \frac{8}{60} A_k \nonumber \\
          & = \frac{6}{15} B_k + \frac{6}{15} A_k, \tag{S21} \label{eq:b_k}
\end{align}

\begin{equation}
  \label{eq:a_k_def}
  \begin{array}{c}
  \includegraphics[width=0.7\textwidth]{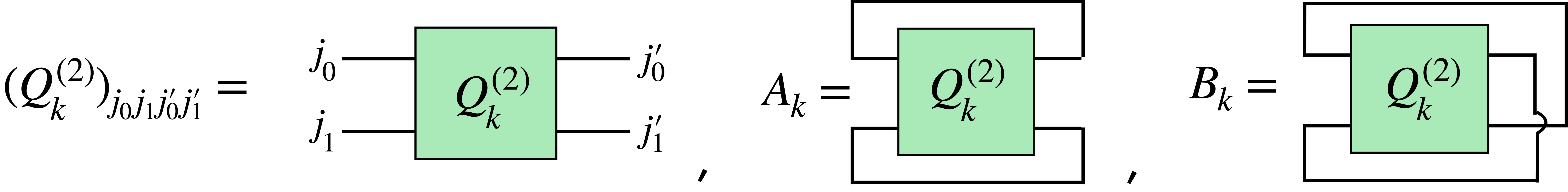}\hspace{3pt}, \tag{S22} \\[6pt]
  \end{array}
\end{equation}

We now give an interpretation of $A_k$ and $B_k$ appearing in \cref{eq:a_k,eq:b_k}. 
We denote $\Delta \rho_k$ as the reduced density matrix after the $k$-th unitary gate in~\cref{eq:delta_rho}.
Then, we identify $A_k$ with $\mathbb{E}[\mathrm{Tr}(\Delta \rho_k)^2]$, namely the integral of the square of the expected trace of $\Delta \rho_k$. 
Similarly, $B_k$ corresponds to $\mathbb{E}[\mathrm{Tr}(\Delta \rho_k^2)]$.
In this light, the fact that coefficient of $A_k$ is 1 in~\cref{eq:a_k} reflects the trace-preserving nature of the quantum gates. 
Moreover, since our initial state is taken to be $Z$, we have $A_0 = 0$ and $B_0 = 1$ at the very beginning.
Consequently, $A_{L-1} = 0$ and $B_{L-1} = {(\frac{6}{15})}^{L - 1}$.
As the variance of the $(j_0, j_0^\prime)$ element of $\Delta \rho$ is given by $(Q^{(2)}_{L})_{j_0, j_0, j_0^\prime, j_0^\prime}$,
and the $Q^{(2)}_k$ is given by linear combination of $A_k$ and $B_k$,
this result immediately proves that $\mathrm{Var}[\Delta \rho]$ decays as $O({(\frac{6}{15})}^{L})$ for the system size $L$.

\end{proof}

\begin{lemma}\label{lemma:grad_deltarho}
  Expectation value of the gradient of $\Delta \rho$ vanishes, and the variance decays exponentially as the system size $L$ increases,
  \begin{align}
    \mathbb{E}\left[\frac{\partial}{\partial U_k} \Delta \rho\right] &= \mathbf{0}, \tag{S23} \\
    \mathrm{Var}\left[\frac{\partial}{\partial U_k} \Delta \rho\right] &= O((\frac{6}{15})^{L}) \tag{S24}
  \end{align}
  Here $\frac{\partial}{\partial U_k}$ denotes the Euclidean derivative, treating 
  $U_k$  as a complex-valued  $4 \times 4$ matrix (as opposed to a point on the unitary group).
\end{lemma}

\begin{proof}
We consider to take the derivative of $\Delta \rho$ with respect to $(U_k)_{ab}$.
Then a term such as~\cref{eq:grad_rho} appears.

\begin{equation}
  \label{eq:grad_rho}
  \begin{array}{c}
  \includegraphics[width=0.7\textwidth]{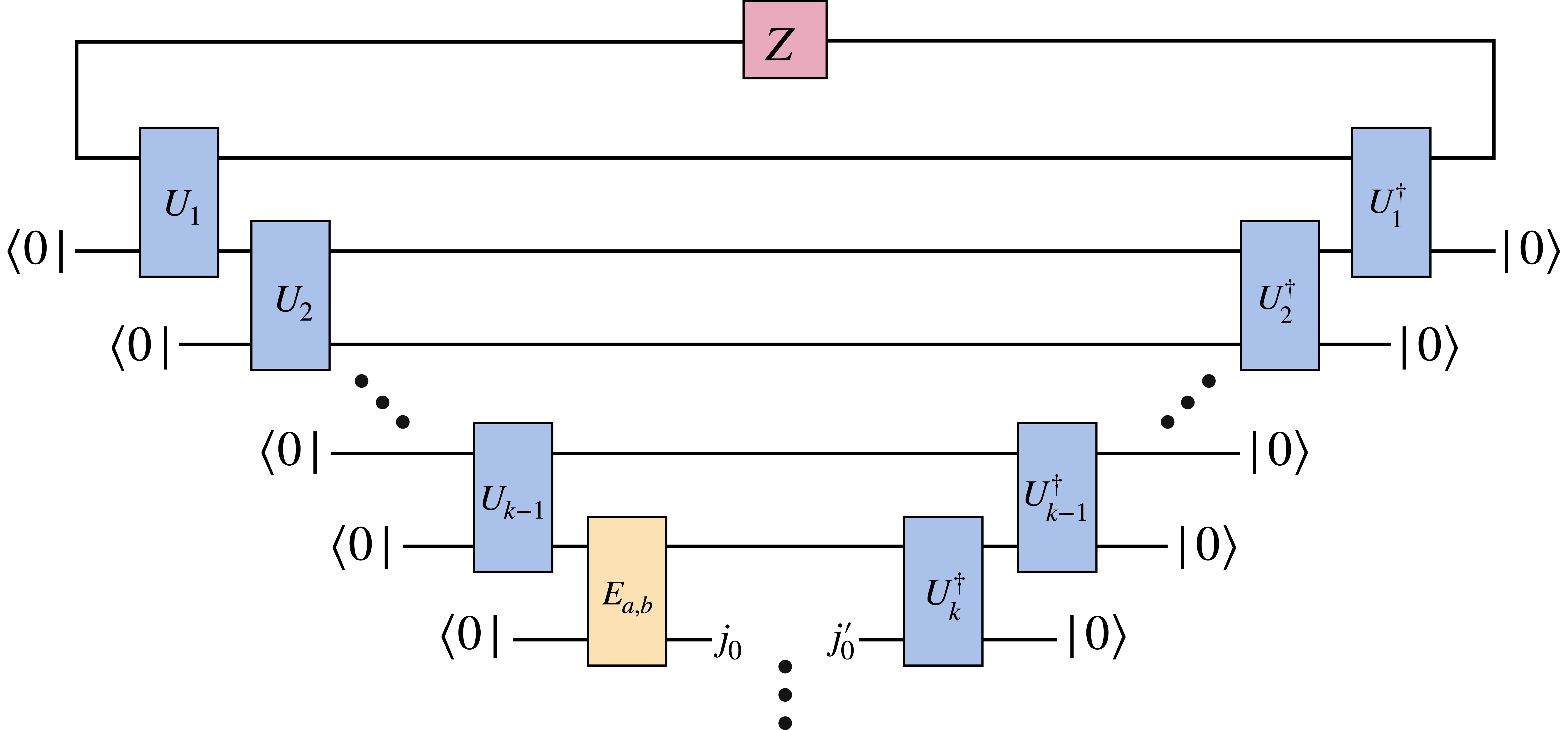}\hspace{3pt}, \tag{S25} \\[6pt]
  \end{array}
\end{equation}

Here, $E_{ab}$ is the unit matrix where the $(a,b)$ element is one and all other elements are zero.
By integrating~\cref{eq:grad_rho} over the Haar measure, and noting that $\int_{U(N)} dU_H \, U = 0$, one straightforwardly confirms that
$
\mathbb{E}[\frac{\partial}{\partial U_k} \Delta \rho] = \mathbf{0}.
$

It's also straightforward to show that 
$\mathrm{Var}[\frac{\partial}{\partial U_k} \Delta \rho] = O((\frac{6}{15})^{L})$, 
by using the same discussion in the proof of~\cref{lemma:mean-deltarho}.
We want to calculate $\mathbb{E}[{(\frac{\partial}{\partial (U_k)_{ab}} \Delta \rho)}^{\odot 2}]$.
Let us first consider the case $k = L-1$, i.e., the last unitary gate. 
In this situation, we can integrate over all unitary gates up to $L-2$ using \cref{eq:a_k,eq:b_k,eq:a_k_def};
since the final integration over $U_{L-1}$ does not alter the overall order, we conclude that 
$
\mathbb{E}[{(\frac{\partial}{\partial {(U_{L-1})}_{ab}} \Delta \rho)}^{\odot 2}] = O((\frac{6}{15})^{L}).
$
Conversely, for $k=1$, the integral over the subsequent unitary gates ($2,\ldots,L-1$) may again be evaluated via \cref{eq:a_k,eq:b_k,eq:a_k_def}. 
One might worry that a finite value of $A_{1}$ could propagate according to 
Eq.~(\ref{eq:a_k}), but such concern is circumvented by the fact that a CPTP map 
is trace-preserving, implying $A_{1} = 0$. 
Hence, we again find that 
$
\mathbb{E}[{(\frac{\partial}{\partial {(U_{1})}_{ab}} \Delta \rho)}^{\odot 2}] = O((\frac{6}{15})^{L}).
$
Applying the same reasoning to any unitary gate $k$ establishes 
$
\mathrm{Var}[\frac{\partial}{\partial U_{k}} \Delta \rho] = O((\frac{6}{15})^{L}).
$
\end{proof}

By using~\cref{lemma:mean-deltarho}, ~\cref{lemma:grad_deltarho} and chebyshev's inequality, we have
\begin{align}
  \Pr \Big( \big\| \Delta \rho_L(x_{2:L})
  \big\| & < \epsilon \Big) = \frac{1}{\epsilon^2} \mathcal{O} \Big( (\frac{6}{15})^{L} \Big), \tag{S26} \\
  \Pr \Big( \big\| \frac{\partial \Delta \rho_L(x_{2:L})}{\partial U_k}
  \big\| & < \epsilon \Big) = \frac{1}{\epsilon^2} \mathcal{O} \Big( (\frac{6}{15})^{L} \Big), \tag{S27}
\end{align}

Hence, we proved~\cref{thm:exp_decay}.

Next, we show that under this condition, training inevitably stalls 
for any choice of loss function. Consider a dataset $\mathcal{D}$ 
and define its loss $\mathcal{L}(\mathcal{D})$ by
\begin{equation}
    \mathcal{L}(\mathcal{D}) \;=\; \frac{1}{N_{\mathcal{D}}}\,
    \sum_{(\vec{x}_i,\,l_i)\,\in\,\mathcal{D}}\,\mathcal{L}\bigl(\vec{x}_i,\;l_i\bigr) \tag{S28},
\end{equation}
where $N_{\mathcal{D}}$ is the number of samples in $\mathcal{D}$, 
and $\mathcal{L}(\vec{x}_i, l_i)$ denotes the loss for each individual data sample 
$(\vec{x}_i, l_i)$. In the case of the “First-Qubit Trigger” dataset, 
the loss takes the form
\begin{align*}
  \label{eq:loss_first_qubit_trigger}
  \mathcal{L}(\mathcal{D}_{fq}) 
  &= \frac{1}{N_{\mathcal{D}_{fq}}}\,\sum_{x_{2:L},\,l}\,
  \mathcal{L}\bigl(\rho_L(l,\,x_{2:L}),\;l\bigr) \nonumber \\[4pt]
  &= \frac{1}{N_{\mathcal{D}_{fq}}}\,\sum_{x_{2:L}}
  \left[\mathcal{L}\bigl(\rho_L(0,\,x_{2:L}),\,0\bigr)\;+\;
  \mathcal{L}\bigl(\rho_L(1,\,x_{2:L}),\,1\bigr)\right] \tag{S29} \\
  &= \frac{1}{N_{\mathcal{D}_{fq}}}\,\sum_{x_{2:L}}
  \left[\mathcal{L}\bigl(\rho_L(0,\,x_{2:L}),\,0\bigr)\;+\;
  \mathcal{L}\bigl(\rho_L(0,\,x_{2:L}),\,1\bigr)\right]\;+\;
  \frac{\partial \mathcal{L}}{\partial \rho} \Delta \rho \tag{S30}.
\end{align*}
where $\rho_L(l, x_{2:L})$ is defined in~\cref{thm:exp_decay} and 
we used the fact that $\rho_L(0, x_{2:L}) - \rho_L(1, x_{2:L}) = O((\frac{6}{15})^L)$.
We can reasonably assume that the effect of the term 
$\mathcal{L}(\rho_L(0, x_{2:L}), 0) + \mathcal{L}(\rho_L(0, x_{2:L}), 1)$
never contribute to distinguishing $\rho_L(0, x_{2:L})$ from $\rho_L(1, x_{2:L})$,
as the term merely enforce the output from the qMPS-classifiers to be (0.5, 0.5),
which is the balanced outcome.
Even worse, this term only makes the training harder.
Thus, no matter which loss function is employed, the exponential decay of the gradient of $\Delta \rho_L$ 
ensures that training become exponentially hard as $L$ grows. 

\begin{figure}[htp]
  \centering
  \subfloat[]{\includegraphics[width=0.5\textwidth]{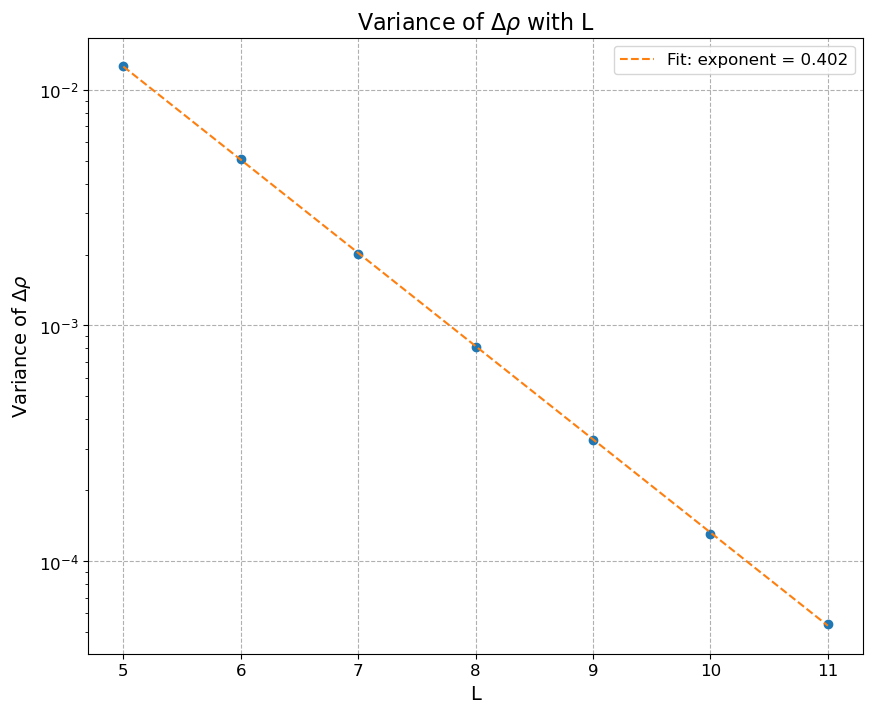}}
  \subfloat[]{\includegraphics[width=0.5\textwidth]{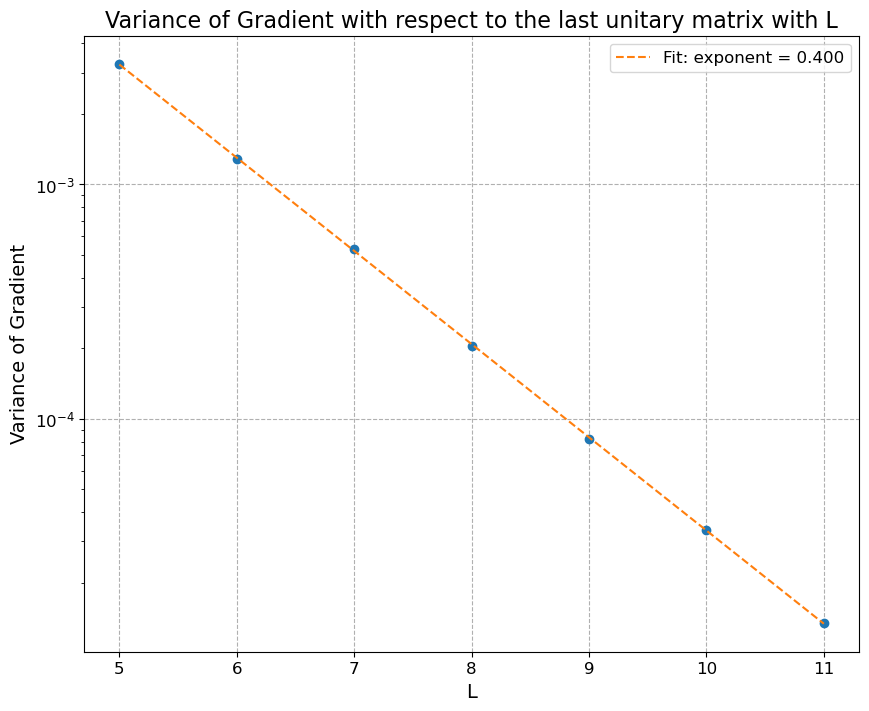}}
  \caption{(a) Graph of $\Delta\rho_L$ vs. $L$. (b) Graph of $\frac{\partial \Delta\rho_L}{\partial U_L}$ vs. $L$.
  We calculated the variance and mean by sampling 5000 random initialization.
  }\label{fig:bp_graph}
\end{figure}

\cref{fig:bp_graph} confirms numerically that 
both $\Delta\rho_L$ and $\frac{\partial \Delta\rho_L}{\partial U_L}$ 
scale as $O((\frac{6}{15})^L)$.

\subsection{Proof of Theorem 2: Absence of Barren Plateaus in Classical MPS}
As outlined in Theorem~\cref{thm:mps_training} and~\cref{sssec:supp_mps_impl}, we assume 
that the MPS is initialized by stacked identity with small random noise as in~\cref{eq:mps_init}. 
Using a tensor-network representation, the function 
\(f_{\ell}(\vec{x}) = W^{\ell}\cdot \phi(\vec{x})\) 
can be expressed as in~\cref{eq:mps_tn}.

\begin{equation}
  \label{eq:mps_tn}
  \begin{array}{c}
  \includegraphics[width=0.4\textwidth]{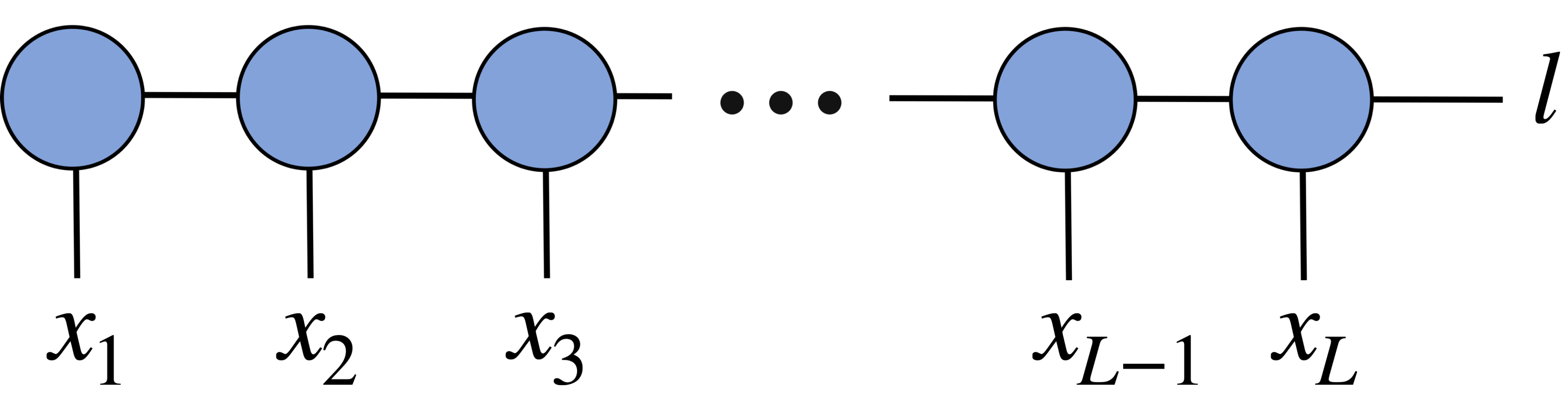}\hspace{3pt}, \tag{S31} \\[6pt]
  \end{array}
\end{equation}
Here, $\phi(x) = (1 - x,\, x)$, and for any input $x_{2:L}$ 
we define the matrix $M(x_{2:L})$ via~\cref{eq:M},
\begin{equation}
  \label{eq:M}
  \begin{array}{c}
  \includegraphics[width=0.6\textwidth]{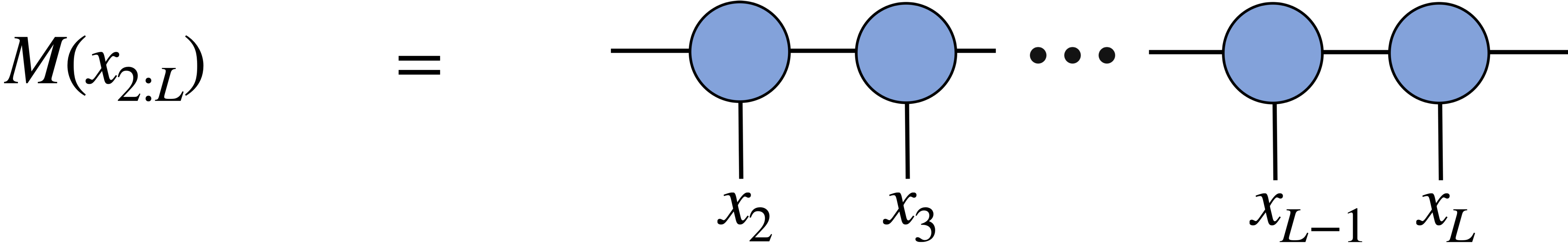}\hspace{3pt}, \tag{S32} \\[6pt]
  \end{array}
\end{equation}

Recalling our stacked identity initialization of the MPS, it follows that
\begin{equation*}
M(x_{2:L}) 
= \mathbb{I} + N(x_{2:L}) + o({(\sqrt{L-1}\epsilon)}), \tag{S33}
\end{equation*}
where 
\begin{equation*}
  N(x_{2:L}) 
  = \sum_{i=2}^L N^{x_i}. \tag{S34}
\end{equation*}
This behaves like a normal distribution with mean 0 and variance $(L-1)\epsilon^2$, which means the magnitude of $N(x_{2:L})$ is $O(\sqrt{L-1}\epsilon)$.
One can set $\epsilon = \epsilon' / \sqrt{L}$ such that $\epsilon^\prime << 1$ and drop the $o(\sqrt{L-1}\epsilon)$ term and we obtain the simple leading form
\begin{equation*}
  M(x_{2:L}) \approx \begin{pmatrix}
    1 +  N_{11}(x_{2:L}) & N_{12}(x_{2:L}) \\
    N_{21}(x_{2:L}) & 1 + N_{22}(x_{2:L}).
  \end{pmatrix} \tag{S35}
\end{equation*}
Here $N_{ij}(x_{2:L})$ is the $(i,j)$-th element of $N(x_{2:L})$.
When class label is binary i.e. $\ell \in \{0, 1\}$, by multiply the leftmost tensor with 
$M(x_{2:L})$, we obtain the MPS output as

\begin{equation*}
  \label{eq:fl}
  f_{l}(\vec{x}) = 
  \begin{pmatrix}
    1-x_1 & x_1
  \end{pmatrix}
  \begin{pmatrix}
    1 + a & 1 + b \\
    1 + c & 1 + d
  \end{pmatrix}
  \begin{pmatrix}
    1 +  N_{11}(x_{2:L}) &  N_{12}(x_{2:L}) \\
     N_{21}(x_{2:L}) & 1 +  N_{22}(x_{2:L})
  \end{pmatrix}
  \begin{pmatrix}
    1 - l \\
    l
  \end{pmatrix}. \tag{S36}
\end{equation*}
Here, the matrix $A^{s_1}_{i_1}$ is expressed as~\cref{eq:A1}. 

\begin{equation*}
  \label{eq:A1}
  A^{s_1}_{i_1} = 
  \begin{pmatrix}
    1 + a & 1 + b \\
    1 + c & 1 + d
  \end{pmatrix} 
  , \tag{S37}
\end{equation*}
where $a,b,c,d$ are drawn from a normal distribution with mean 0 and variance $\epsilon^2$.

~\cref{eq:fl} allows us to interpret $f_{l}(\vec{x})$ as $(x_1, l)$ element of matrix $F$ with the relation $f_{l}(\vec{x}) = F(x_{2:L})_{x_1 l}$,
where $F(x_{2:L})$ is given by

\begin{equation*}
  \label{eq:fl_leading}
  F(x_{2:L}) = 
  \begin{pmatrix}
    1 + a +  N_{11} + N_{21}  & 1 + b + N_{12} + N_{22} \\
    1 + c +  N_{11} + N_{21}  & 1 + d + N_{12} + N_{22}
  \end{pmatrix}. \tag{S38}
\end{equation*}

As in the main text, we consider to apply the softmax activation function, followed by the NLL to calculate the loss function

\begin{align}
  \label{eq:loss_mps}
  \mathcal{L}(\mathcal{D}_{fq}) &= - \frac{1}{N_{\mathcal{D}_{fq}}} \sum_{\vec{x}_i, l_i} \mathcal{L}(\vec{x}_i, l_i) \nonumber \\
  &= - \frac{1}{N_{\mathcal{D}_{fq}}} \sum_{(\vec{x}_i, l_i) \in \mathcal{D}_{fq}} 
  \log \frac{e^{f_{l_i}(\vec{x}_i)}}{e^{f_{l_i}(\vec{x}_i)} + e^{f_{1-l_i}(\vec{x}_i)}} \nonumber \\
  &= - \frac{1}{N_{\mathcal{D}_{fq}}} \sum_{x_{2:L}}  \ell_{\mathrm{LS}}{(F(x_{2:L})_{x_1, :})}_0 +  \ell_{\mathrm{LS}}{(F(x_{2:L})_{x_1, :})}_1. \tag{S39}
\end{align}

Here, $\ell_{\mathrm{LS}}(F(x_{2:L})_{x_1, :})$ represents the log-softmax. 
Note that $F(x_{2:L})_{x_1, :} = (F(x_{2:L})_{x_1, 0}, F(x_{2:L})_{x_1, 1})$ is a 2-dimensional vector
and $\ell_{\mathrm{LS}}(\cdot)$ can be seen as the function that takes a 2-dimensional vector and returns a 2-dimensional vector.

We use the following expansion of the log-softmax function to first order in $\Delta f$,
\begin{equation*}
\begin{pmatrix}
 \ell_{\mathrm{LS}}(1+\Delta f_1, 1+\Delta f_2)_0\\
 \ell_{\mathrm{LS}}(1+\Delta f_1, 1+\Delta f_2)_1
\end{pmatrix}
\;=\;
\begin{pmatrix}
 \log\frac{1}{2} \\[4pt]
 \log\frac{1}{2}
\end{pmatrix}
\;+\;\frac{1}{2}
\begin{pmatrix}
 1 & -1 \\[4pt]
 -1 & 1
\end{pmatrix}
\begin{pmatrix}
 \Delta f_1 \\[3pt]
 \Delta f_2
\end{pmatrix}
+ O(\Delta f^2)
 ,
 \tag{S40}
\end{equation*}
which yields, for the following approximation of~\cref{eq:loss_mps}, 
\begin{equation*}\label{eq:loss_mps_expansion}
- \frac{1}{N_{\mathcal{D}_{fq}}}\sum_{x_{2:L}} \Bigl[ \log(\frac{1}{2}) + \frac{1}{2}(b+c-a-d)\Bigr] + O(\epsilon^\prime) 
= - \frac{1}{4}(b+c-a-d) + O(\epsilon^\prime) + const. \tag{S38}
\end{equation*}
Note that we have $2^{L-1}$ summations in $\sum_{x_{2:L}}$ and $N_{\mathcal{D}_{fq}} = 2^{L-1}$.
Hence, the partial derivative with respect to the left-most tensor becomes
\(\bigl|\tfrac{\partial\mathcal{L}}{\partial A^{s_1}_{i_1}}\bigr|\approx \tfrac14 + O(\epsilon')\).
This implies that the MPS-classifiers initialized with a stacked identity does not does exhibit barren plateaus for this First-Qubit-Trigger dataset. 
While this alone suffices to show the absence of barren plateaus, one can further check that 
\(\bigl\|\tfrac{\partial \mathcal{L}}{\partial A^{s_k}_{i_k, i_{k+1}}}\bigr\|\! = O(\epsilon')\)
for \(k>1\), which does not scale with $L$ if one set $\epsilon = \epsilon' / \sqrt{L}$ with $\epsilon^\prime$ being small constant such as $0.1$.
We can prove this straightforwardly by expanding the log-likelihood to second order in $\epsilon'$, since the effect of 
$A^{s_k}_{i_k, i_{k+1}}$ does not appear in the first order expansion~\cref{eq:loss_mps_expansion}, 
but we omit the detail here.
Both are numerically verified in~\cref{fig:no_bp}.
Intuitively, this result is consistent with the fact that, unlike a qMPS whose sequential unitaries progressively erase the information of the prior-layer,
MPS-classifiers using a near-identity initialization keep the information of the prior-layer and 
propagate it forward.
As discussed in the main text, absence of barren plateaus does not guarantee the MPS-classifiers can find the optimal parameters for the First-Qubit-Trigger dataset.
In~\cref{sec:numerical_verification}, we numerically verify that the MPS-classifiers can indeed find the optimal parameters for the First-Qubit-Trigger dataset.
Also some readers may wonder what if one initialize the qMPS-classifiers by embedding the MPS-classifiers with stacked identity initialization with random noise.
We will discuss this as well in~\cref{sec:numerical_verification}.


\begin{figure}[htp]
  \centering
  \subfloat[]{\includegraphics[width=0.5\textwidth]{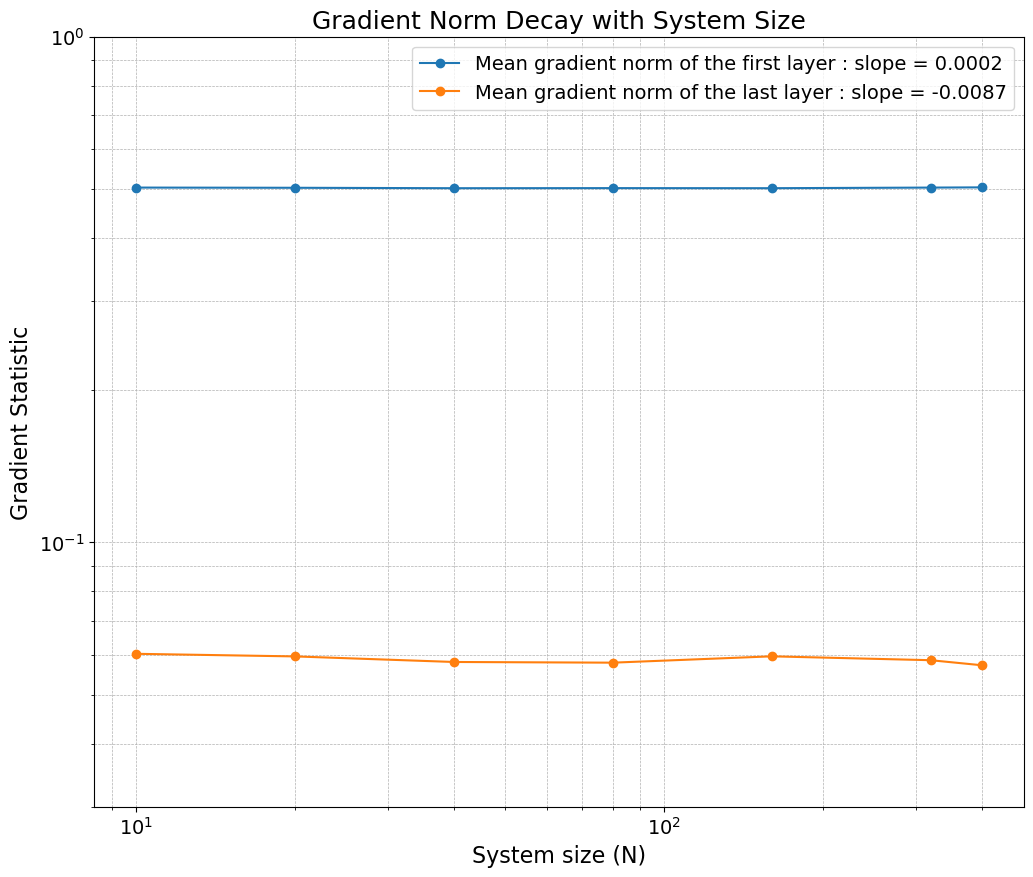}}
  \caption{The gradients of the loss function do not scale with the system size $L$ for the MPS-classifier with stacked identity initialization.
  We used $\epsilon^\prime = 0.1$ and run 1000 independent trials to estimate the norm of the gradient of the loss function for different system sizes $L$.
  }\label{fig:no_bp}
\end{figure}

\subsection{Proof of Theorem 3: Exponential Post-Selection Cost}
In what follows, we show that embedding MPS-classifiers capable of perfectly 
classifying the first-qubit trigger dataset into a PQC inevitably requires 
post-selection with an exponentially low success probability. We begin by 
discussing the general aspects of post-selection. We express the MPS-classifiers
in the right-canonical form as in~\cref{eq:mps-canonical}. 
\begin{equation}\label{eq:mps-canonical}
  \begin{array}{c}
  \includegraphics[width=0.6\textwidth]{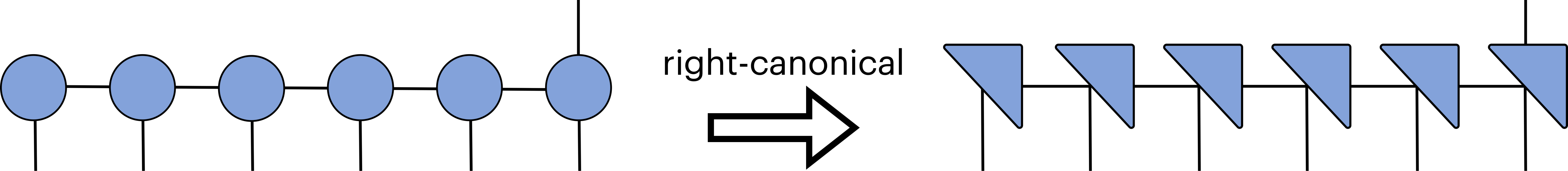}\hspace{3pt}. \tag{S39} \\[6pt]
  \end{array}
\end{equation}

We omit the canonical center in~\cref{eq:mps-canonical}
as it will be absorbed into the 
measurement operator $\hat{M}$ when embedding into quantum circuits and to not affect the post-selection.
Under these conditions, the tensor $W$ can be 
expressed as
\begin{equation*}\label{eq:mps-state} 
W = \lvert 0\rangle\langle W_0 \rvert + \lvert 1\rangle\langle W_1 \rvert. \tag{S40}
\end{equation*} 
Note that from~\cref{eq:mps-canonical}, the Frobenius norm of $W$ is $\sum_\ell = 2$.

\begin{lemma}\label{lemma:post_selection} 
When one embeds the MPS-classifiers of~\cref{eq:mps-state} into qMPS 
circuits exactly, the required post-selection success probability for any input 
$\lvert \phi \rangle$ is $\sum_\ell \|\langle W_\ell \lvert \psi \rangle\|^2$. 

\begin{proof}
First, we embed the MPS-classifiers in right-canonical form into qMPS 
circuits with the following~\cref{eq:mps-pqc}
\begin{equation}\label{eq:mps-pqc}
  \begin{array}{c}
  \includegraphics[width=0.6\textwidth]{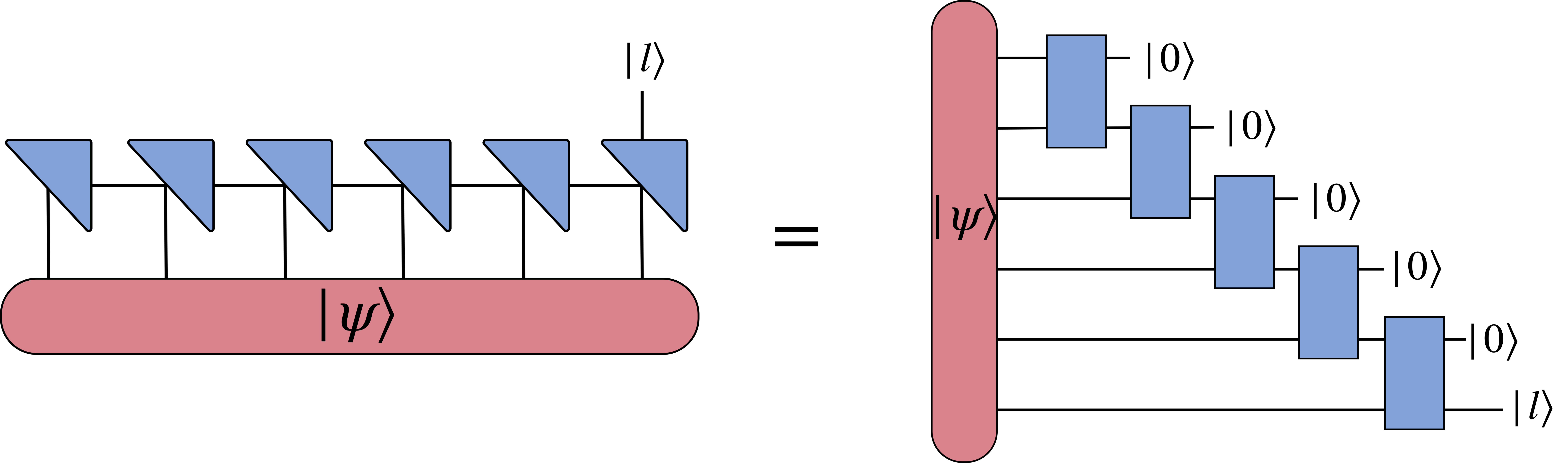}\hspace{3pt}. \tag{S41} \\[6pt]
  \end{array}.
\end{equation}

By noting that the success probability $P_{\text{success}}$ is given by
\begin{equation*} 
P_{\text{success}} = \sum_\ell \bigl\lVert \langle 0,0,\dots,0,\ell \vert 
U_{\text{qMPS}} \lvert \psi\rangle \bigr\rVert^2 , \tag{S42}
\end{equation*}
we can obtain that 
\begin{equation*} 
P_{\text{success}} = \sum_\ell | \langle \ell | W^\ell | \psi \rangle |^2 = \sum_\ell \|\langle W_\ell \lvert \psi \rangle\|^2.
\end{equation*}
\end{proof}
\end{lemma}
Suppose that for each input $\vec{x}_i$ in the dataset $D$, the embedded quantum states $\phi(\vec{x}_i)$ are perpendicular to each other, i.e., 
\begin{equation}
\langle \phi(\vec{x}_i) \mid \phi(\vec{x}_j)\rangle = 0 
\quad \text{for} \quad i \neq j.
\end{equation}
Let $D_\ell \subseteq D$ be the subset of data points with class label $\ell$, and then 
the following $W^\ell$ is the optimal tensor-network classifiers in terms of the training dataset.
\begin{equation}
\lvert W_\ell\rangle 
=
\frac{1}{\sqrt{N_{\mathcal{D}_\ell}}}
\sum_{\vec{x}\,\in\,D_\ell}
\lvert \phi(\vec{x})\rangle,
\end{equation}
The above discussion shows that any architecture of TN-classifiers will be trained to approximate $W$ with a given TN structure.
For the first-qubit trigger dataset, one can construct MPS-classifiers that exactly represent this $W$ with the parameters in~\cref{lemma:perfect_classifier}.
Therefore, by the~\cref{lemma:post_selection}, the post-selection success probability $P_{\text{success}}$ for any bit string $\vec{x}$ in $D_{fq}$ then equals
\begin{equation}
P_{\text{success}} = \sum_\ell \|\langle W_\ell \lvert \phi(\vec{x}) \rangle\|^2 = 2^{-(L - 1)}.
\end{equation}

\subsection{Numerical Verification of Trainability of MPS-classifiers and qMPS-classifiers}
\label{sec:numerical_verification}

We first present numerical results of optimizing a MPS-classifier on the First-Qubit-Trigger dataset with $L=1024$.
~\cref{fig:trainable_fq} displays the loss function and accuracy during training. 
Evidently, the MPS-classifier find parameters that perfectly classify this dataset 
without becoming trapped in local minima.

\begin{figure}[htp]
  \centering
  \includegraphics[width=0.5\textwidth]{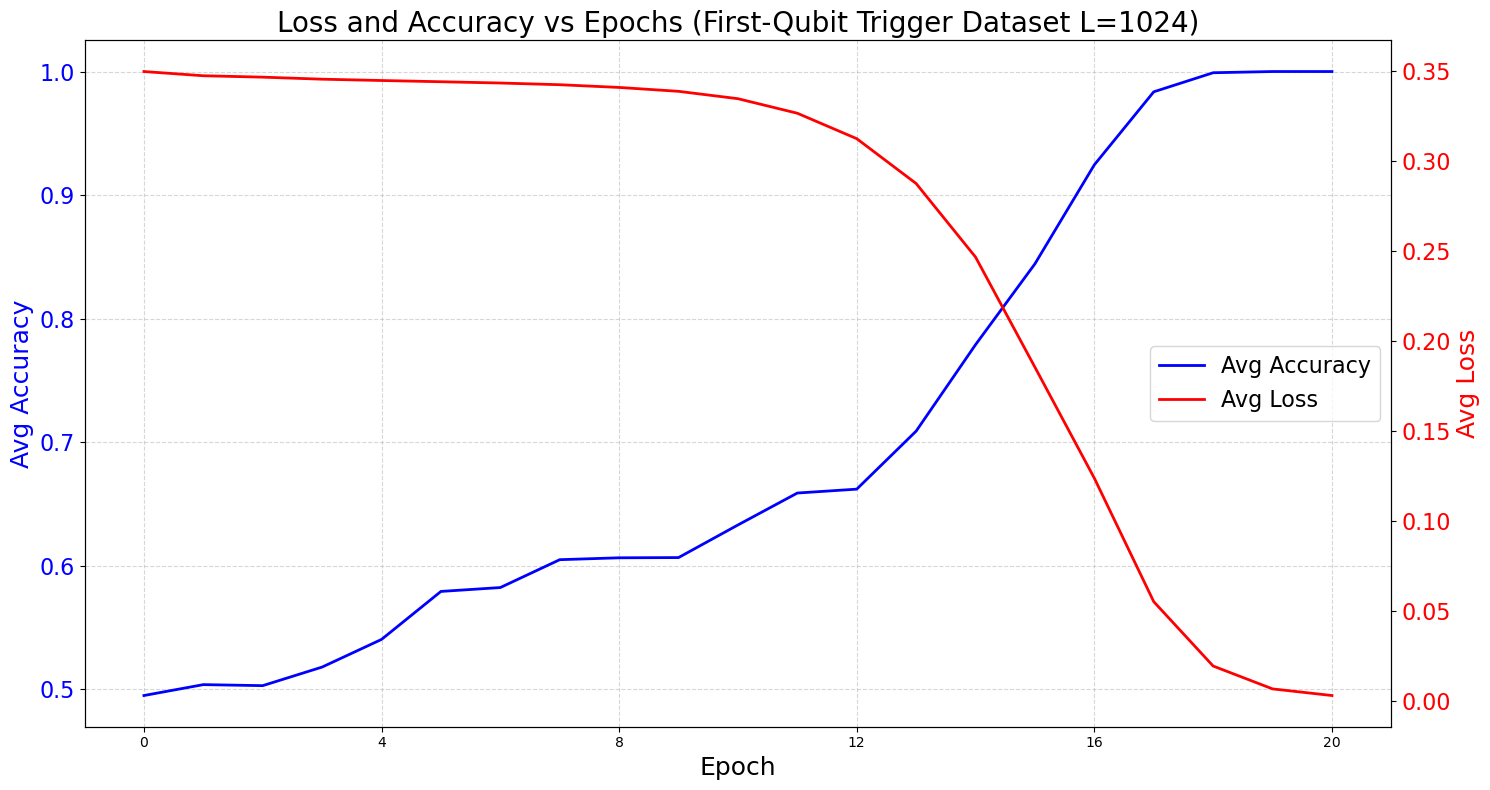}
  \caption{The loss function and accuracy of the MPS-classifier for the First-Qubit-Trigger dataset with $L=1024$.}
  \label{fig:trainable_fq}
\end{figure}

\begin{figure}[htp]
  \centering
  \subfloat[]{\includegraphics[width=0.5\textwidth]{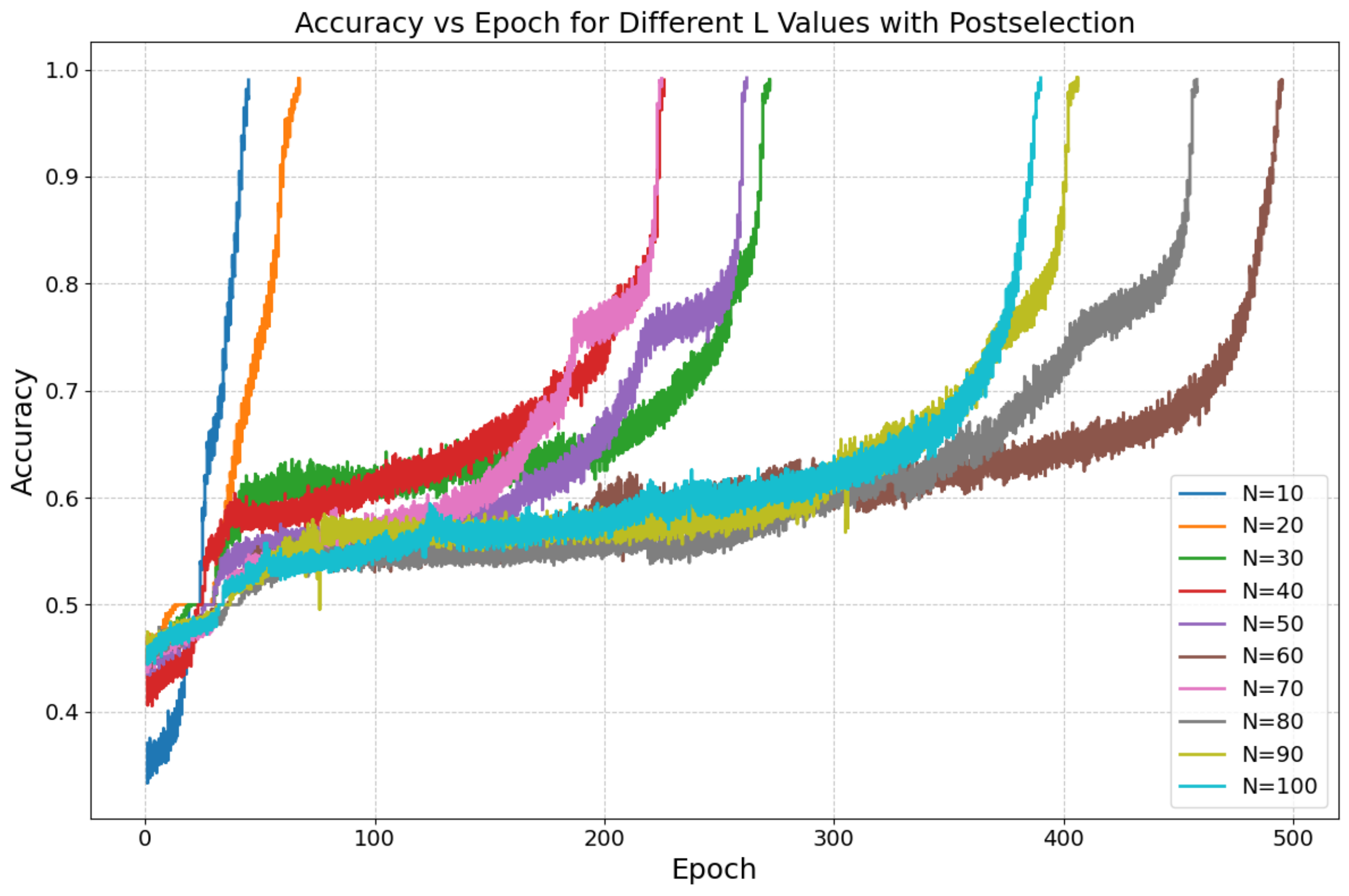}}
  \subfloat[]{\includegraphics[width=0.5\textwidth]{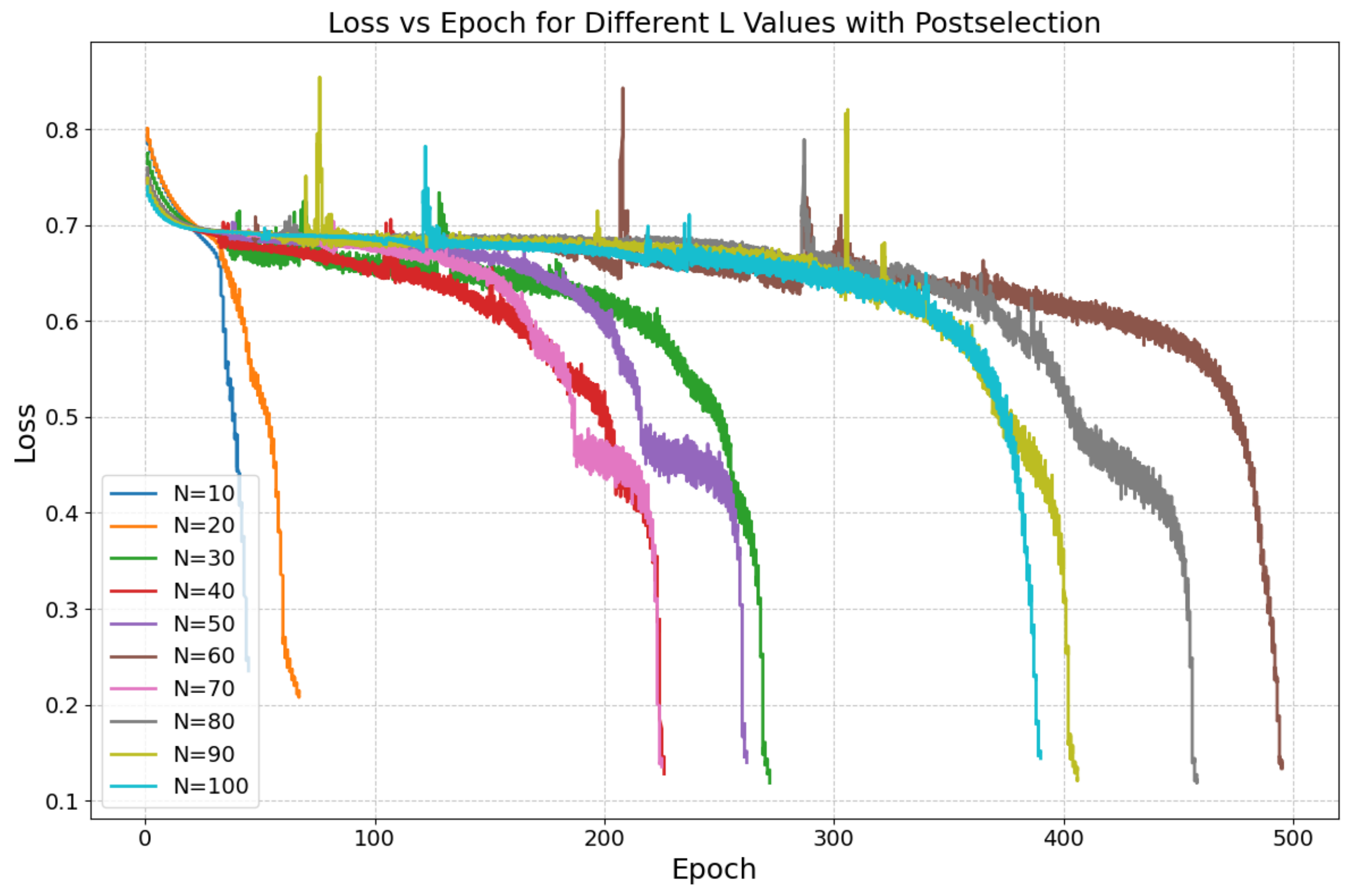}}
  \caption{(a) The accuracy and (b) the loss function of the qMPS-classifier for the First-Qubit-Trigger dataset with the size $L=10, 20, \dots, 100$.}
  \label{fig:qMPSWP_fq}
\end{figure}

We next consider qPMS-classifiers (with and without postselection) 
initialized by embedding the MPS-classifiers with stacked identity initialization plus random noise.
~\cref{fig:qMPSWP_fq} shows that when qMPS-classifiers with postselection is initialized in this way, 
the model can actually find the optimal parameters for the First-Qubit-Trigger dataset.
This is not surprising as qMPS-classifiers with postselection is essentially the same as the MPS-classifiers.
However, while the MPS-classifiers require only about 20 epochs of training for $L=1024$, 
the qMPS-classifiers with postselection needs roughly 400 epochs even for $L=100$. 
This indicates that forcing the parameters to be unitary or isometric significantly 
complicates the optimization, in addition to the higher computational cost. 
In practice, there is thus reason to prefer postselection-based qMPS-classifiers
over direct MPS-classifiers, since the latter can be embedded exactly into the former.

\begin{figure}[htp]
  \centering
  \includegraphics[width=0.7\textwidth]{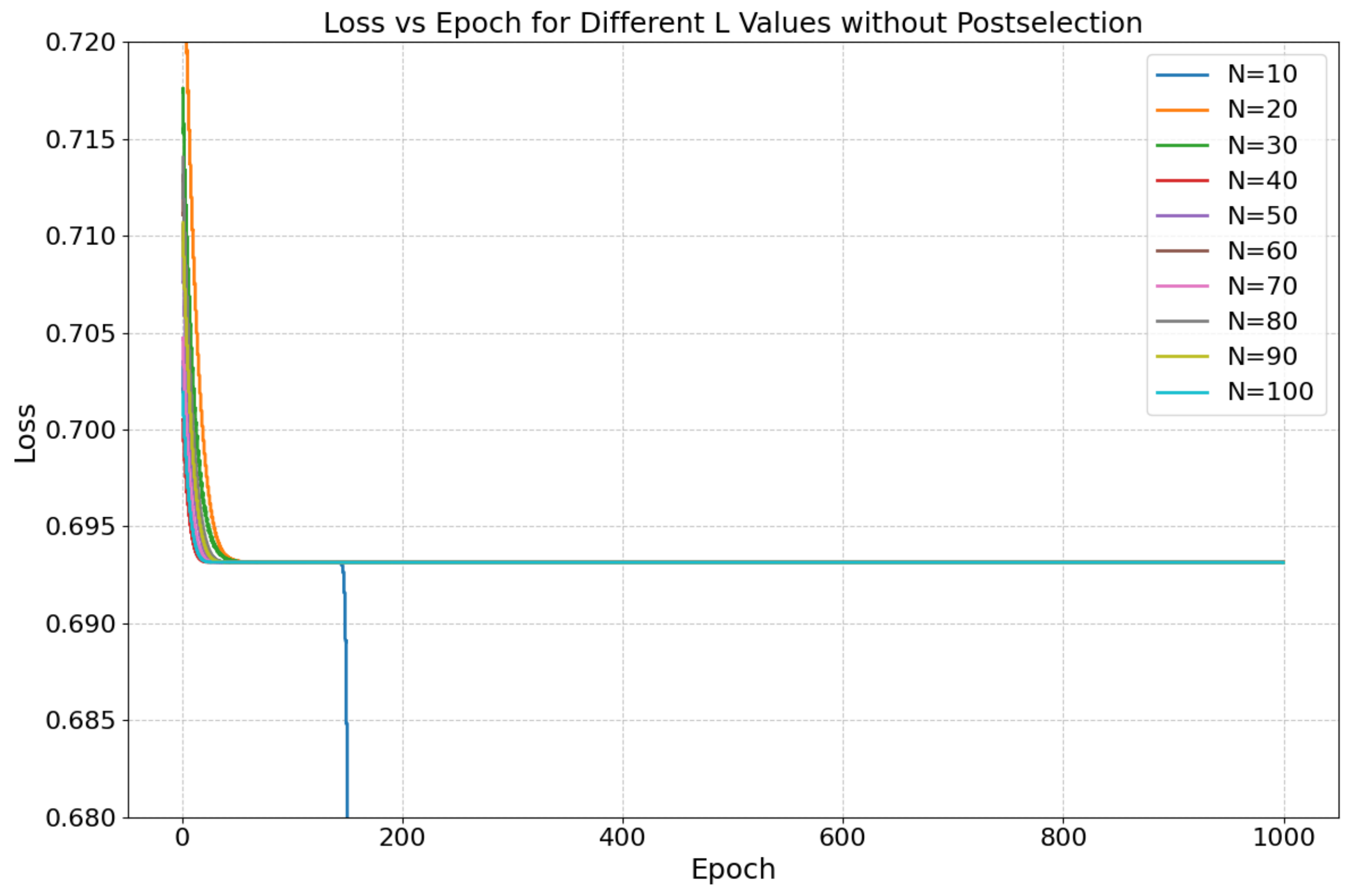}
  \caption{Loss function of the qMPS-classifier \emph{without} postselection on the First-Qubit-Trigger dataset 
  (initialization identical to the postselection case). Only the $L=10$ model escapes the trivial solution at $\approx 0.693$.}
  \label{fig:qMPSWOP_fq}
\end{figure}
Finally,~\cref{fig:qMPSWOP_fq} shows the results for the qMPS-classifier \emph{without} postselection, 
initialized in the same manner. In this case, all models with $L>10$ remain trapped in the trivial solution 
at loss $-\log(0.5)=0.693$, producing a fifty-fifty output $(0.5,\,0.5)$ for any input. 
Hence, even when not initialized with Haar-random unitaries, qMPS-classifiers without postselection 
exhibits exponential training barriers, symptomatic of barren plateaus.


\end{document}